\documentclass{article}[11pt]
\usepackage{amsmath,amssymb,amsfonts,amsthm,bm,array,mathtools}
\usepackage{algorithm,algorithmic,url}
\usepackage{graphicx}
\usepackage[margin=1in]{geometry}
\newtheorem{theorem}{Theorem}
\newtheorem{lemma}[theorem]{Lemma}

\newcommand{\TNorm}[1]{\mbox{}\|#1\|_2}

\newcommand{\abs }[1]{\left|#1\right|}
\newcommand{\entropy}[1]{\ensuremath{\mathcal{H}\left(#1\right)}}
\newcommand{\aentropy}[1]{\ensuremath{\widehat{\mathcal{H}}}\left(#1\right)}
\newcommand{\trace}[1]{\ensuremath{\mathrm{\textbf{\rm \bf tr}}\left(#1\right)}}

\newcommand{\nnz}{\mathrm{nnz}}

\newcommand{\bA}{\mathbf{A}}
\newcommand{\be}{\mathbf{e}}

\newcommand{\bg}{\mathbf{g}}
\newcommand{\bD}{\mathbf{D}}
\newcommand{\bE}{\mathbf{E}}

\newcommand{\bU}{\mathbf{U}}

\newcommand{\bx}{\mathbf{x}}

\newcommand{\by}{\mathbf{y}}

\newcommand{\bI}{\mathbf{I}}
\newcommand{\bB}{\mathbf{B}}

\newcommand{\bC}{\mathbf{C}}
\newcommand{\bH}{\mathbf{H}}

\newcommand{\bR}{\mathbf{R}}
\newcommand{\bS}{\mathbf{S}}

\newcommand{\bSigma}{\mathbf{\Sigma}}

\newcommand{\bPi}{\mathbf{\Pi}}
\newcommand{\bPsi}{\mathbf{\Psi}}
\newcommand{\bpsi}{\boldsymbol{\psi}}

\newcommand{\bzero}{\mathbf{0}}
\newcommand{\real}{\mathbb{R}}
\newcommand{\comp}{\mathbb{C}}

\newcommand{\OO}{\mathcal{O}}

\def\ceil#1{{\left\lceil\,#1\,\right\rceil}} 
\begin{document}

\title{Randomized Linear Algebra Approaches to Estimate the Von Neumann Entropy of Density Matrices}
\author{Eugenia-Maria Kontopoulou \thanks{Purdue University. West Lafayette, IN. Email: ekontopo@purdue.edu.}\and
Gregory-Paul Dexter \thanks{Purdue University. West Lafayette, IN. Email: gdexter@purdue.edu.} \and
Wojciech Szpankowski \thanks{Purdue University. West Lafayette, IN. Email: spa@cs.purdue.edu.}\and
Ananth Grama\thanks{Purdue University. West Lafayette, IN. Email: ayg@cs.purdue.edu.} \and
Petros Drineas \thanks{Purdue University. West Lafayette, IN. Email: pdrineas@purdue.edu.}}
\maketitle

\begin{abstract}
 The \textit{von Neumann entropy}, named after John von Neumann, is an extension of the classical concept of entropy to the field of quantum mechanics. From a numerical perspective, von Neumann entropy can be computed simply by computing all eigenvalues of a density matrix, an operation that could be prohibitively expensive for large-scale density matrices. We present and analyze three randomized algorithms to approximate von Neumann entropy of {real} density matrices: our algorithms leverage recent developments in the Randomized Numerical Linear Algebra (RandNLA) literature, such as randomized trace estimators, provable bounds for the power method, and the use of random projections to approximate the eigenvalues of a matrix. All three algorithms come with provable accuracy guarantees and our experimental evaluations support our theoretical findings showing considerable speedup with small loss in accuracy.
\end{abstract}

\section{Introduction}

Entropy is a fundamental quantity in many areas of science and engineering. \textit{von Neumann entropy}, named after John von Neumann, is an extension of classical entropy concepts to the field of quantum mechanics. Its foundations can be traced to von Neumann's work on \textit{Mathematische Grundlagen der Quantenmechanik}\footnote{Originally published in German in 1932; published in English under the title \textit{Mathematical Foundations of Quantum Mechanics} in 1955.}. In his work, Von Neumann introduced the notion of a \textit{density matrix}, which facilitated extension of the tools of classical statistical mechanics to the quantum domain in order to develop a theory of quantum mechanics.

From a mathematical perspective (see Section~\ref{sxn:notation} for details) the {real} density matrix $\bR$ is a symmetric positive semidefinite matrix in $\mathbb{R}^{n \times n}$ with unit trace. Let $p_i$, $i=1\ldots n$ be the eigenvalues of $\bR$ in decreasing order; then, the entropy of $\bR$ is defined as\footnote{$\bR$ is symmetric positive semidefinite and thus all its eigenvalues are non-negative. If $p_i$ is equal to zero we set $p_i \ln p_i$ to zero as well.}
\begin{equation}\label{eqn:earlyR}
\mathcal{H}(\bR) = -\sum_{i=1}^n p_i \ln p_i.
\end{equation}
The above definition is a proper extension of both the Gibbs entropy and the Shannon entropy to the quantum case. It implies an obvious algorithm to compute $\mathcal{H}(\bR)$ by computing the eigendecomposition of $\bR$; known algorithms for this task can be prohibitively expensive for large values of $n$, particularly when the matrix becomes dense~\cite{GVL96}. For example,~\cite{Wihler2014} describes an entangled two-photon state generated by spontaneous parametric down-conversion, which can result in a {sparse and banded} density matrix with $n\approx 10^8$.

Motivated by the high computational cost, we seek numerical algorithms that approximate the von Neumann entropy of large density matrices, e.g., symmetric positive definite matrices with unit trace, faster than the trivial $\OO(n^3)$ approach. Our algorithms build upon recent developments in the field of Randomized Numerical Linear Algebra (RandNLA), an interdisciplinary research area that exploits randomization as a computational resource to develop improved algorithms for large-scale linear algebra problems. Indeed, our work here focuses at the intersection of RandNLA and information theory, delivering novel randomized linear algebra algorithms and related quality-of-approximation results for a fundamental information-theoretic metric.

\subsection{Background}\label{sxn:notation}
We focus on finite-dimensional function (state) spaces. In this setting, the density matrix $\bR$ represents the statistical mixture of $k\leq n$ pure states, and has the form
\begin{equation}\label{eqn:R}
\bR = \sum_{i=1}^k p_i \bpsi_i \bpsi_i^T  \in \mathbb{R}^{n \times n}.
\end{equation}
The vectors $\bpsi_i \in \mathbb{R}^n$ for $i=1\ldots k$ represent the $k \leq n$ pure states and can be assumed to be pairwise orthogonal and normal, while $p_i$'s correspond to the probability of each state and satisfy $p_i > 0$ and $\sum_{i=1}^k p_i = 1$. From a linear algebraic perspective, eqn.~(\ref{eqn:R}) can be rewritten as
\begin{equation}\label{eqn:R1}
\bR = \bPsi \bSigma_p \bPsi^T \in \mathbb{R}^{n \times n},
\end{equation}
where $\bPsi \in \mathbb{R}^{n \times k}$ is the matrix whose columns are the vectors $\bpsi_i$ and $\bSigma_p \in \mathbb{R}^{k \times k}$ is a diagonal matrix whose entries are the (positive) $p_i$'s. Given our assumptions for $\bpsi_i$, $\bPsi^T\bPsi=\bI$; also $\bR$ is symmetric positive semidefinite with its eigenvalues equal to $p_i$ and corresponding left/ right singular vectors equal to $\bpsi_i$'s; and $\trace{\bR}=\sum_{i=1}^k p_i = 1$. Notice that eqn.~(\ref{eqn:R1}) essentially reveals the (thin) Singular Value Decomposition (SVD)~\cite{GVL96} of $\bR$. The Von Neumann entropy of $\bR$, denoted by $\mathcal{H}(\bR)$ is equal to (see also eqn.~(\ref{eqn:earlyR}))
\begin{equation}\label{eqn:VN}
\mathcal{H}(\bR) = -\sum_{i: p_i> 0} p_i \ln p_i = -\trace{\bR \ln \bR}.
\end{equation}
The second equality follows from the definition of matrix functions~\cite{Higham:2008}. More precisely, we overload notation and consider the full SVD of $\bR$, namely $\bR = \bPsi \bSigma_p \bPsi^T$, where $\bPsi \in \mathbb{R}^{n \times n}$ is an orthogonal matrix whose top $k$ columns correspond to the $k$ pure states and the bottom $n-k$ columns are chosen so that $\bPsi \bPsi^T = \bPsi^T \bPsi = \bI_n$. Here $\bSigma_p$ is a diagonal matrix whose bottom $n-k$ diagonal entries are set to zero. Let $h(x) = x \ln x$ for any $x>0$ and let $h(0)=0$. Then, using the cyclical property of the trace and the definition of $h(x)$,
\begin{eqnarray}
\nonumber-\sum_{i, p_i> 0} p_i \ln p_i
\nonumber &=& -\trace{\bPsi h(\bSigma_{p}) \bPsi^T}\\
\nonumber &=& -\trace{h(\bR)}\\
&=& -\trace{\bR \ln \bR}.
\end{eqnarray}

\subsection{Trace estimators}

The following lemma appeared in~\cite{Boutsidis2016} and is immediate from Theorem~5.2 in~\cite{AT11}. It implies an algorithm to approximate the trace of any symmetric positive semidefinite matrix $\bA$ by computing inner products of the matrix with Gaussian random vectors.
\begin{lemma}
	\label{thm:trace}
	Let $\bA \in \mathbb{R}^{n \times n}$ be a positive semi-definite matrix, let $0<\epsilon < 1$ be an accuracy parameter, and let $0<\delta<1$ be a failure probability. If $\bg_1,\bg_2,\ldots, \bg_s \in\mathbb{R}^n$ are independent random standard Gaussian vectors, then, for $s= \left\lceil 20 \ln(2/\delta) / \epsilon^2 \right \rceil$, with probability at least $1 - \delta$,
	\begin{equation*}\label{eq:trApprox}
	\abs{
		\trace{\bA} - \frac1{s} \sum_{i=1}^{s} \bg_i^\top \bA \bg_i
	} \leq \epsilon \cdot \trace{\bA}.
	\end{equation*}
\end{lemma}
\subsection{Our contributions}

We present and analyze three randomized algorithms to approximate the von Neumann entropy of density matrices. The \textit{first two} algorithms (Sections~\ref{sxn:Taylor} and~\ref{sxn:cheb}) leverage two different polynomial approximations of the matrix function  $\mathcal{H}(\bR) = -\trace{\bR \ln \bR}$: the first approximation uses a Taylor series expansion, while the second approximation uses Chebyschev polynomials. Both algorithms return, with high probability, relative-error approximations to the true entropy of the input density matrix, under certain assumptions. More specifically, in both cases, we need to assume that the input density matrix has $n$ non-zero eigenvalues, or, equivalently, that the probabilities $p_i$, $i=1\ldots n$, corresponding to the underlying $n$ pure states are non-zero. The running time of both algorithms is proportional to the sparsity of the input density matrix and depends (see Theorems~\ref{thm:taylor} and~\ref{thm:cheb} for precise statements) on, roughly, the ratio of the largest to the smallest probability $p_1/p_n$ (recall that the smallest probability is assumed to be non-zero), as well as the desired accuracy.

The \textit{third} algorithm (Section~\ref{sxn:rp}) is fundamentally different, if not orthogonal, to the previous two approaches. It leverages the power of random projections~\cite{Drineas2016,Woodruff2014} to approximate numerical linear algebra quantities, such as the eigenvalues of a matrix. Assuming that the density matrix $\bR$ has exactly $k \ll n$ non-zero eigenvalues, e.g., there are $k$ pure states with non-zero probabilities $p_i$, $i=1\ldots k$, the proposed algorithm returns, with high probability, relative error approximations to all $k$ probabilities $p_i$. This, in turn, implies an additive-relative error approximation to the entropy of the density matrix, which, under a mild assumption on the true entropy of the density matrix, becomes a relative error approximation (see Theorem~\ref{thm:rp} for a precise statement). The running time of the algorithm is again proportional to the sparsity of the density matrix and depends on the target accuracy, but, unlike the previous two algorithms, does not depend on any function of the $p_i$.

From a technical perspective, the theoretical analysis of the first two algorithms proceeds by combining the power of polynomial approximations, either using Taylor series or Chebyschev polynomials, to matrix functions, combined with randomized trace estimators. A provably accurate variant of the power method is used to estimate the largest probability $p_1$. If this estimate is significantly smaller than one, it can improve the running times of the proposed algorithms (see discussion after Theorem~\ref{thm:taylor}). The third algorithm leverages a powerful, multiplicative matrix perturbation result that first appeared in~\cite{Demmel1992}. Our work in Section~\ref{sxn:rp} is a novel application of this inequality to derive bounds for RandNLA algorithms.

Finally, in Section~\ref{sxn:experiments}, we present a detailed evaluation of our algorithms on synthetic density matrices of various sizes, most of which were generated using Matlab's QETLAB toolbox~\cite{qetlab}. For some of the larger matrices that were used in our evaluations, the exact computation of the entropy takes hours, whereas our algorithms return approximations with relative errors well below $0.5\%$ in only a few minutes.

\subsection{Prior work}

The first non-trivial algorithm to approximate the von Neumann entropy of a density matrix appeared in~\cite{Wihler2014}. Their approach is essentially the same as our approach in Section~\ref{sxn:cheb}. Indeed, our algorithm in Section~\ref{sxn:cheb} was inspired by their approach. However, our analysis is somewhat different, leveraging a provably accurate variant of the power method, as well as provably accurate trace estimators to derive a relative error approximation to the entropy of a density matrix, under appropriate assumptions. A detailed, technical comparison between our results in Section~\ref{sxn:cheb} and the work of~\cite{Wihler2014} is delegated to Section~\ref{sxn:wihler}.

Independently and in parallel with our work,~\cite{MNSUW18} presented a multipoint interpolation algorithm (building upon~\cite{HNO08}) to compute a relative error approximation for the entropy of a real matrix with bounded condition number. The proposed running time of Theorem 35 of~\cite{MNSUW18} does not depend on the condition number of the input matrix (i.e., the ratio of the largest to the smallest probability), which is a clear advantage in the case of ill-conditioned matrices. However, the dependence of the algorithm of Theorem 35 of~\cite{MNSUW18} on terms like $(\log{n}/\epsilon)^6$ or $n^{1/3} \nnz(\bA)+ n \sqrt{\nnz(\bA)}$ (where $\nnz(\bA)$ represents the number of non-zero elements of the matrix $\bA$) could blow up the running time of the proposed algorithm for reasonably conditioned matrices.

We also note the recent work in~\cite{Boutsidis2016}, which used Taylor approximations to matrix functions to estimate the log determinant of symmetric positive definite matrices (see also Section 1.2 of~\cite{Boutsidis2016} for an overview of prior work on approximating matrix functions via Taylor series). The work of~\cite{Han2015} used a Chebyschev polynomial approximation to estimate the log determinant of a matrix and is reminiscent of our approach in Section~\ref{sxn:cheb} and, of course, the work of~\cite{Wihler2014}.

We conclude this section by noting that our algorithms use two tools (described, for the sake of completeness, in the Appendix) that appeared in prior work. The first tool is the power method, with a provable analysis that first appeared in~\cite{LT_Lecture}. The second tool is a provably accurate trace estimation algorithm for symmetric positive semidefinite matrices that appeared in~\cite{AT11}. 

\section{An approach via Taylor series}\label{sxn:Taylor}

Our first approach to approximate the von Neumann entropy of a density matrix uses a Taylor series expansion to approximate the logarithm of a matrix, combined with a relative-error trace estimator for symmetric positive semi-definite matrices and the power method to upper bound the largest singular value of a matrix.

\subsection{Algorithm and Main Theorem}
Our main result is an analysis of Algorithm~\ref{alg1a} (see below) that guarantees relative error approximation to the entropy of the density matrix $\bR$, under the assumption that $\bR = \sum_{i=1}^n p_i \bpsi_i \bpsi_i^T  \in \mathbb{R}^{n \times n}$ has $n$ pure states with $0 < \ell \leq p_i$ for all $i=1\ldots n$.
\begin{algorithm}[h!]
\begin{algorithmic}[1]
\STATE {\bf{INPUT}}: $\bR \in \mathbb{R}^{n \times n}$, accuracy parameter $\varepsilon > 0$, failure probability $\delta$, and integer $m >0$.
\STATE Compute $\tilde{p}_1$, the estimate of the largest eigenvalue of $\bR$, $p_1$, using Algorithm~\ref{alg:power} (see Appendix) with $t=\OO(\ln n)$ and $q=\OO(\ln(1/\delta))$.
\STATE Set $u=\min\{1,6\tilde{p}_1\}$.
\STATE Set $s = \left\lceil 20 \ln(2/\delta) / \varepsilon^2 \right \rceil$.
\STATE Let $\bg_1,\bg_2,\ldots, \bg_s \in \mathbb{R}^n$ be i.i.d. random Gaussian vectors.
\STATE \textbf{OUTPUT}: return
$$\aentropy{\bR} = \ln u^{-1} + \frac{1}{s}\sum_{i=1}^{s}\sum_{k=1}^{m}\frac{\bg_i^\top \bR (\bI_n-u^{-1}\bR)^k \bg_i}{k}.$$
\end{algorithmic}
\caption{{A Taylor series approach to estimate the entropy.}}
\label{alg1a}
\end{algorithm}
The following theorem is our main quality-of-approximation result for Algorithm~\ref{alg1a}.
\begin{theorem}
	\label{thm:taylor}
Let $\bR$ be a density matrix such that all probabilities $p_i$, $i=1\ldots n$ satisfy $0 < \ell \leq p_i$. Let $u$ be computed as in Algorithm~\ref{alg1a} and let $\aentropy{\bR}$ be the output of Algorithm~\ref{alg1a} on inputs $\bR$, $m$, and $\epsilon<1$;
Then, with probability at least $1-2\delta$,
$$
\abs{ \aentropy{\bR} - \entropy{\bR} }  \leq
2\epsilon \entropy{\bR},$$
by setting $m=\left\lceil\frac{u}{\ell}\ln{\frac{1}{\epsilon}}\right\rceil$. The algorithm runs in time
$$\OO\left(\left(\frac{u}{\ell}\cdot\frac{\ln(1/\epsilon)}{\epsilon^2}+\ln(n)\right)\ln(1/\delta)\cdot \nnz(\bR)\right).$$ %
\end{theorem}
A few remarks are necessary to better understand the above theorem. First, $\ell$ could be set to $p_n$, the smallest of the probabilities corresponding to the $n$ pure states of the density matrix $\bR$. Second, it should be obvious that $u$ in Algorithm~\ref{alg1a} could be simply set to one and thus we could avoid calling Algorithm~\ref{alg:power} to estimate $p_1$ by $\tilde{p}_1$ and thus compute $u$. However, if $p_1$ is small, then $u$ could be significantly smaller than one, thus reducing the running time of Algorithm~\ref{alg1a}, which depends on the ratio $u/\ell$. Third, ideally, if both $p_1$ and $p_n$ were used instead of $u$ and $\ell$, respectively, the running time of the algorithm would scale with the ratio $p_1/p_n$.

\subsection{Proof of Theorem~\ref{thm:taylor}}

We now prove Theorem~\ref{thm:taylor}, which analyzes the performance of Algorithm~\ref{alg1a}. Our first lemma presents a simple expression for $\entropy{\bR}$ using a Taylor series expansion.
\begin{lemma}\label{lem1a}
Let $\bR \in \mathbb{R}^{n \times n}$ be a symmetric positive definite matrix with unit trace and whose eigenvalues lie in the interval $[\ell,u]$, for some $0< \ell \leq u \leq 1$. Then,
\begin{equation*}
\entropy{\bR} = \ln u^{-1} +  \sum_{k=1}^{\infty} \frac{\trace{\bR(\bI_n - u^{-1}\bR)^k}}{k}.
\end{equation*}
\end{lemma}

\begin{proof}
From the definition of the von Neumann entropy and a Taylor expansion,
\begin{eqnarray}
\nonumber\entropy{\bR}&=&
-\trace{\bR \ln \left(uu^{-1}\bR\right)}\\
\nonumber &=& -\trace{(\ln u)\bR} -\trace{\bR \ln(\bI_n - (\bI_n - u^{-1}\bR))}\\
\label{eqn:PD1}&=& \ln u^{-1}-\trace{ -\bR\sum_{k=1}^{\infty}  \frac{(\bI_n - u^{-1}\bR)^k}{k}} \\
\nonumber&=& \ln u^{-1}+\sum_{k=1}^{\infty} \frac{\trace{\bR (\bI_n - u^{-1}\bR)^k}}{k}.
\end{eqnarray}
Eqn.~(\ref{eqn:PD1}) follows since $\bR$ has unit trace and from a Taylor expansion: indeed, $\ln(\bI_n - \bA) = -\sum_{k=1}^\infty \bA^k/k$ for a symmetric matrix $\bA$ whose eigenvalues are all in the interval $(-1,1)$. We note that the eigenvalues of $\bI_n-u^{-1}\bR$ are in the interval $[0,1-(\ell/u)]$, whose upper bound is strictly less than one since, by our assumptions, $\ell/u>0$.
\end{proof}

We now proceed to prove Theorem~\ref{thm:taylor}. We will condition our analysis on Algorithm~\ref{alg:power} being successful, which happens with probability at least $1-\delta$. In this case, $u = \min\{1,6\tilde{p}_1\}$ is an upper bound for all probabilities $p_i$.
For notational convenience, set $\bC = \bI_n-u^{-1}\bR$. We start by manipulating $\Delta=\abs{\aentropy{\bR} -
		\entropy{\bR}}$ as follows:
\begin{align*}
\Delta &=  \abs{  \sum_{k=1}^{m}\frac{1}{k} \cdot \frac{1}{s} \sum_{i=1}^{s}\bg_i^\top \bR\bC^k \bg_i  - \sum_{k=1}^{\infty}\frac{1}{k} \trace{\bR\bC^k}}    \\
&\le\abs{  \sum_{k=1}^{m}\frac{1}{k}\cdot\frac{1}{s} \sum_{i=1}^{s}    \bg_i^\top \bR\bC^k \bg_i-\sum_{k=1}^{m}\frac{1}{k} \trace{\bR\bC^k} }+\abs{ \sum_{k=m+1}^{\infty}\frac{1}{k} \trace{\bR\bC^k}} \\
&=\underbrace{\abs{  \frac{1}{s} \sum_{i=1}^{s} \bg_i^\top \left(\sum_{k=1}^{m} \bR\bC^k/k  \right) \bg_i  - \trace{\sum_{k=1}^{m}\frac{1}{k} \bR\bC^k }}}_{\Delta_1}+\underbrace{\abs{ \sum_{k=m+1}^{\infty} \trace{\bR\bC^k}/k}}_{\Delta_2}.
\end{align*}
	\noindent We now bound the two terms $\Delta_1$ and $\Delta_2$ separately. We start with $\Delta_1$:
	the idea is to apply Lemma~\ref{thm:trace} on the matrix $\sum_{k=1}^{m} \bR\bC^k/k$ with
	$s = \left\lceil 20 \ln(2/\delta) / \epsilon^2 \right \rceil$. Hence, with probability at least $1-\delta$:
\begin{equation}\label{eqn:d1}
	\Delta_1 \le \epsilon \cdot \trace{\sum_{k=1}^{m} \bR\bC^k/k} \le \epsilon \cdot \trace{\sum_{k=1}^{\infty} \bR\bC^k/k}.
\end{equation}
A subtle point in applying Lemma~\ref{thm:trace} is that the matrix $\sum_{k=1}^{m} \bR\bC^k/k$ must be symmetric positive semidefinite. To prove this, let the SVD of $\bR$ be $\bR = \bPsi \bSigma_p \bPsi^T$, where all three matrices are in $\mathbb{R}^{n \times n}$ and the diagonal entries of $\bSigma_p$ are in the interval $[\ell,u]$. Then, it is easy to see that $\bC = \bI_n - u^{-1}\bR = \bPsi(\bI_n-u^{-1}\bSigma_p)\bPsi^T$ and $\bR\bC^k = \bPsi\bSigma_p(\bI_n-u^{-1}\bSigma_p)^k\bPsi^T$, where the diagonal entries of $\bI_n-u^{-1}\bSigma_p$ are non-negative, since the largest entry in $\bSigma_p$ is upper bounded by $u$. This proves that $\bR\bC^k$ is symmetric positive semidefinite for any $k$, a fact which will be useful throughout the proof. Now,
$$\sum_{k=1}^m \bR\bC^k/k = \bPsi \left(\bSigma_p \sum_{k=1}^m (\bI_n-u^{-1}\bSigma_p)^k/k\right)\bPsi^T,$$
which shows that the matrix of interest is symmetric positive semidefinite. Additionally, since $\bR\bC^k$
is symmetric positive semidefinite, its trace is non-negative, which proves the second inequality in eqn.~(\ref{eqn:d1}) as well.

We proceed to bound $\Delta_2$ as follows:
\begin{align}
\nonumber \Delta_2 &= \abs{\sum_{k=m+1}^{\infty} \trace{\bR\bC^k}/k }
= \abs{\sum_{k=m+1}^{\infty} \trace{\bR\bC^m \bC^{k-m}}/k}\\
\label{eqn:PD21}&= {\abs{\sum_{k=m+1}^{\infty} \trace{\bC^m  \bC^{k-m}\bR}/k}
	\le \abs{\sum_{k=m+1}^{\infty} \TNorm{\bC^m} \cdot \trace{\bC^{k-m}\bR}/k}}\\
\label{eqn:PD22}&= {\TNorm{\bC^m} \cdot \abs{\sum_{k=m+1}^{\infty}\trace{ \bR\bC^{k-m}}/k}
	\le \TNorm{\bC^{m}} \cdot \abs{\sum_{k=1}^{\infty} \trace{\bR\bC^k}/k }}\\
\label{eqn:PD24} &\le \left(1-\frac{\ell}{u}\right)^m \sum_{k=1}^{\infty} \trace{\bR\bC^k}/k.
\end{align}
To prove eqn.~(\ref{eqn:PD21}), we used von Neumann's trace inequality\footnote{Indeed, for any two matrices $\bA$ and $\bB$, $\trace{\bA\bB} \leq \sum_i \sigma_i(\bA)\sigma_i(\bB)$, where $\sigma_i(\bA)$ (respectively $\sigma_i(\bB)$) denotes the $i$-th singular value of $\bA$ (respectively $\bB$). Let $\TNorm{\cdot}$ to denote the induced-2 matrix or spectral norm, then $\TNorm{\bA} = \sigma_1(\bA)$ (its largest singular value). Given that each singular value of $\bA$ is upper bounded by $\sigma_1(\bA)$ then we can rewrite $\trace{\bA\bB} \leq \TNorm{\bA}\sum_i \sigma_i(\bB)$; if $\bB$ is symmetric positive semidefinite, $\trace{\bB}=\sum_i \sigma_i(\bB)$.}. Eqn.~\eqref{eqn:PD21} now follows since $\bC^{k-m}\bR$ is symmetric positive semidefinite\footnote{This can be proven using an argument similar to the one used to prove eqn.~(\ref{eqn:d1}).}. To prove eqn.~(\ref{eqn:PD22}), we used the fact that $\trace{\bR\bC^k}/k\geq 0$ for any $k\geq 1$.
Finally, to prove eqn.~(\ref{eqn:PD24}), we used the fact that $\TNorm{\bC} = \TNorm{\bI_n-u^{-1}\bSigma_p} \leq 1-\ell/u$ since the smallest entry in $\Sigma_p$ is at least $\ell$ by our assumptions. We also removed unnecessary absolute values since $\trace{\bR\bC^k}/k$ is non-negative for any positive integer $k$.

Combining the bounds for $\Delta_1$ and $\Delta_2$ gives
	\begin{align*}
	\abs{ \aentropy{\bR} - \entropy{\bR} } & \leq 
	\left(\epsilon+\left(1-\frac{\ell}{u}\right)^m\right) \sum_{k=1}^{\infty} \frac{\trace{\bR\bC^k}}{k}.
	\end{align*}
	\noindent We have already proven in Lemma~\ref{lem1a} that
	$$ \sum_{k=1}^{\infty} \frac{\trace{\bR\bC^k}}{k} \leq \entropy{\bR} - \ln u^{-1} \leq \entropy{\bR},$$
where the last inequality follows since $u\leq 1$. Collecting our results, we get
		\begin{align*}
		\abs{ \aentropy{\bR} - \entropy{\bR} }
		\leq
		\left(\epsilon+\left(1-\frac{\ell}{u}\right)^m\right)  \entropy{\bR}.
		\end{align*}
Setting
$$m=\left\lceil\frac{u}{\ell}\ln{\frac{1}{\epsilon}}\right\rceil$$
and using $\left(1-x^{-1}\right)^x \leq e^{-1}$ ($x>0$),
guarantees that $(1-\ell/u)^m \leq \epsilon$ and concludes the proof of the theorem.
We note that the failure probability of the algorithm is at most $2\delta$ (the sum of the failure probabilities of the power method and the trace estimation algorithm).

Finally, we discuss the running time of Algorithm~\ref{alg1a}, which is equal to $\OO(s\cdot m\cdot \nnz(\bR) )$. Since $s = \OO\left(\frac{\ln(1/\delta)}{\epsilon^{2}}\right)$ and $m=\OO\left(\frac{u\ln(1/\epsilon)}{\ell}\right)$, the running time becomes (after accounting for the running time of Algorithm~\ref{alg:power}) $$\OO\left(\left(\frac{u}{\ell}\cdot\frac{\ln(1/\epsilon)}{\epsilon^2}+\ln(n)\right)\ln(1/\delta)\cdot \nnz(\bR)\right).$$

\section{An approach via Chebyschev polynomials}\label{sxn:cheb}

Our second approach is to use a Chebyschev polynomial-based approximation scheme to estimate the entropy of a density matrix. Our approach follows the work of~\cite{Wihler2014}, but our analysis uses the trace estimators of~\cite{AT11} and Algorithm~\ref{alg:power} and its analysis. Importantly, we present conditions under which the proposed approach is competitive with the approach of Section~\ref{sxn:Taylor}.

\subsection{Algorithm and Main Theorem}

The proposed algorithm leverages the fact that the von Neumann entropy of a density matrix $\bR$ is equal to the (negative) trace of the matrix function $\bR \ln \bR$ and approximates the function $\bR \ln \bR$ by a sum of Chebyschev polynomials; then, the trace of the resulting matrix is estimated using the trace estimator of~\cite{AT11}.

Let $f_m(x) = \sum_{w=0}^m \alpha_w \mathcal{T}_w(x)$ with $\alpha_0 = \frac{u}{2}\left(\ln\frac{u}{4}+1\right)$, $\alpha_1 = \frac{u}{4}\left(2\ln\frac{u}{4}+3\right)$, and $\alpha_w = \frac{(-1)^w u}{w^3-w}$ for $w \geq 2$. Let $\mathcal{T}_w(x) = \cos(w\cdot \arccos ((2/u)x-1))$ and $x \in [0,u]$ be the Chebyschev polynomials of the first kind for any integer $w > 0$. Algorithm~\ref{alg1b} computes $u$ (an upper bound estimate for the largest probability $p_1$ of the density matrix $\bR$) and then computes $f_m(\bR)$ and estimates its trace.
We note that the computation $\bg_i^\top f_m(\bR) \bg_i$ can be done efficiently using Clenshaw's algorithm; see Appendix C for the well-known approach.
\begin{algorithm}[h!]
\begin{algorithmic}[1]
\STATE {\bf{INPUT}}: $\bR \in \mathbb{R}^{n \times n}$, accuracy parameter $\varepsilon > 0$, failure probability $\delta$, and integer $m >0$.
\STATE Compute $\tilde{p}_1$, the estimate of the largest eigenvalue of $\bR$, $p_1$, using Algorithm~\ref{alg:power} (see Appendix) with $t=\OO(\ln n)$ and $q=\OO(\ln(1/\delta))$.
\STATE Set $u=\min\{1,6\tilde{p}_1\}$.
\STATE Set $s = \left\lceil 20 \ln(2/\delta) / \varepsilon^2 \right \rceil$.
\STATE Let $\bg_1,\bg_2,\ldots, \bg_s \in \mathbb{R}^n$ be i.i.d. random Gaussian vectors.
\STATE{\bf{OUTPUT}:
$\aentropy{\bR} = -\frac1{s}\sum_{i=1}^{s} \bg_i^\top f_m(\bR) \bg_i.$
}
\end{algorithmic}
\caption{{A Chebyschev polynomial-based approach to estimate the entropy.}}
\label{alg1b}
\end{algorithm}

\noindent Our main result is an analysis of Algorithm~\ref{alg1b} that guarantees a relative error approximation to the entropy of the density matrix $\bR$, under the assumption that $\bR = \sum_{i=1}^n p_i \bpsi_i \bpsi_i^T  \in \mathbb{R}^{n \times n}$ has $n$ pure states with $0 < \ell \leq p_i$ for all $i=1\ldots n$. The following theorem is our main quality-of-approximation result for Algorithm~\ref{alg1b}.
\begin{theorem}
	\label{thm:cheb}
Let $\bR$ be a density matrix such that all probabilities $p_i$, $i=1\ldots n$ satisfy $0 < \ell \leq p_i$. Let $u$ be computed as in Algorithm~\ref{alg1a} and let $\aentropy{\bR}$ be the output of Algorithm~\ref{alg1b} on inputs $\bR$, $m$, and $\epsilon<1$;
	Then, with probability at least $1-2\delta$,
	$$
	\abs{ \aentropy{\bR} - \entropy{\bR} }  \leq
	3\epsilon \entropy{\bR},$$
by setting  $m = \sqrt{\frac{u}{2\epsilon\ell \ln(1/(1-\ell))}}$. The algorithm runs in time
$$\OO\left(\left(\sqrt{\frac{u}{\ell \ln(1/(1-\ell))}}\cdot\frac{1}{\epsilon^{2.5}}+\ln(n)\right)\ln(1/\delta)\cdot \nnz(\bR)\right).$$
\end{theorem}
The similarities between Theorems~\ref{thm:taylor} and~\ref{thm:cheb} are obvious: same assumptions and directly comparable accuracy guarantees. The only difference is in the running times: the Taylor series approach has a milder dependency on $\epsilon$, while the Chebyschev-based approximation has a milder dependency on the ratio $u/\ell$, which controls the behavior of the probabilities $p_i$. However, for small values of $\ell$ ($\ell \rightarrow 0$),
$$\ln\frac{1}{1-\ell} = \ln\left(1+\frac{\ell}{1-\ell}\right) \approx \frac{\ell}{1-\ell} \approx \ell.$$
Thus, the Chebyschev-based approximation has a milder dependency on $u$ but not necessarily $\ell$ when compared to the Taylor-series approach.  We also note that the discussion following Theorem~\ref{thm:taylor} is again applicable here.

\subsection{Proof of Theorem~\ref{thm:cheb}}

We will condition our analysis on Algorithm~\ref{alg:power} being successful, which happens with probability at least $1-\delta$. In this case, $u = \min\{1,6\tilde{p}_1\}$ is an upper bound for all probabilities $p_i$. We now recall (from Section~\ref{sxn:notation}) the definition of the function $h(x)=x \ln x$ for any real $x \in (0,1]$, with $h(0)=0$. Let $\bR = \bPsi \bSigma_p \bPsi^T \in \mathbb{R}^{n \times n}$ be the density matrix, where both $\bSigma_p$ and $\bPsi$ are matrices in $\mathbb{R}^{n \times n}$. Notice that the diagonal entries of $\bSigma_p$ are the $p_i$s and they satisfy $0 < \ell \leq p_i \leq u \leq 1$ for all $i=1\ldots n$.

Using the definitions of matrix functions from~\cite{Higham:2008}, we can now define $h(\bR) = \bPsi h(\bSigma_p) \bPsi^T$, where $h(\bSigma_p)$ is a diagonal matrix in $\mathbb{R}^{n \times n}$ with entries equal to $h(p_i)$ for all $i=1\ldots n$.
We now restate Proposition 3.1 from~\cite{Wihler2014} in the context of our work, using our notation.
\begin{lemma}
	\label{lem:prior}
The function $h(x)$ in the interval $[0,u]$ can be approximated by
$$f_m(x) = \sum_{w=0}^m \alpha_w \mathcal{T}_w(x),$$
where $\alpha_0 = \frac{u}{2}\left(\ln\frac{u}{4}+1\right)$, $\alpha_1 = \frac{u}{4}\left(2\ln\frac{u}{4}+3\right)$, and  $\alpha_w = \frac{(-1)^w u}{w^3-w}$ for $w \geq 2$.
For any $m\geq 1$, $$\abs{h(x)-f_m(x)}\leq \frac{u}{2m(m+1)}\leq \frac{u}{2m^2},$$ for  $x \in [0,u]$.
\end{lemma}
In the above, $\mathcal{T}_w(x) = \cos(w\cdot \arccos ((2/u)x-1))$ for any integer $w\geq 0$ and $x \in [0,u]$. Notice that the function $(2/u)x-1$ essentially maps the interval $[0,u]$, which is the interval of interest for the function $h(x)$, to $[-1,1]$, which is the interval over which Chebyschev polynomials are commonly defined. The above theorem exploits the fact that the Chebyschev polynomials form an orthonormal basis for the space of functions over the interval $[-1,1]$.

We now move on to approximate the entropy $\mathcal{H}(\bR)$ using the function $f_m(x)$. First,
\begin{eqnarray}\label{eq:cheb_no_trace}
 \nonumber -\trace{f_m(\bR)}&=&{-\trace{\sum_{w=0}^m \alpha_w \mathcal{T}_w(\bR)}
 	= -\trace{\sum_{w=0}^m \alpha_w \bPsi \mathcal{T}_w(\bSigma_p)\bPsi^T}} \\
 \nonumber &=& {-\sum_{w=0}^m \alpha_w \trace{\mathcal{T}_w(\bSigma_p)}= -\sum_{w=0}^m \alpha_w \sum_{i=1}^n \mathcal{T}_w(p_i)}\\
 &=&-\sum_{i=1}^n \sum_{w=0}^m \alpha_w  \mathcal{T}_w(p_i).
\end{eqnarray}

Recall from Section~\ref{sxn:notation} that $\mathcal{H}(\bR)=-\sum_{i=1}^n h(p_i)$.
We can now bound the difference between $\trace{-f_m(\bR)}$ and $\mathcal{H}(\bR)$. Indeed,
\begin{eqnarray}
\nonumber  \abs{\mathcal{H}(\bR)-\trace{-f_m(\bR)}} &=& \abs{-\sum_{i=1}^n h(p_i)+\sum_{i=1}^n \sum_{w=0}^m \alpha_w  \mathcal{T}_w(p_i)} \\
\nonumber &\leq& \sum_{i=1}^n \abs{h(p_i)-\sum_{w=0}^m \alpha_w  \mathcal{T}_w(p_i)}\\
\label{eqn:PD2}&\leq& \frac{nu}{2m^2}.
\end{eqnarray}
The last inequality follows by the final bound in Lemma~\ref{lem:prior}, since all $p_i$'s are in the interval $[0,u]$. 

Recall that we also assumed that all $p_i$s are lower-bounded by $\ell>0$ and thus
\begin{equation}\label{eqn:PD3}
  \mathcal{H}(\bR) = \sum_{i=1}^n p_i \ln \frac{1}{p_i} \geq n \ell \ln \frac{1}{1-\ell}.
\end{equation}
We note that the upper bound on the $p_i$s follows since the smallest $p_i$ is at least $\ell>0$ and thus the largest $p_i$ cannot exceed $1-\ell < 1$. We note that we cannot use the upper bound $u$ in the above formula, since $u$ could be equal to one; $1-\ell$ is always strictly less than one but it cannot be a priori computed (and thus cannot be used in Algorithm~\ref{alg1b}), since $\ell$ is not a priori known. 

We can now restate the bound of eqn.~(\ref{eqn:PD2}) as follows:
\begin{align}
    \nonumber\abs{\mathcal{H}(\bR)-\trace{-f_m(\bR)}} \leq & \frac{u}{2m^2 \ell \ln (1/(1-\ell))} \mathcal{H}(\bR)\\ 
    \label{eqn:PD4} \leq& \epsilon \mathcal{H}(\bR),
\end{align}
where the last inequality follows by setting
\begin{equation}\label{eqn:PD5}
  m = \sqrt{\frac{u}{2\epsilon\ell \ln(1/(1-\ell))}}.
\end{equation}

Next, we argue that the matrix $-f_m(\bR)$ is symmetric positive semidefinite (under our assumptions) and thus one can apply Lemma~\ref{thm:trace} to estimate its trace. We note that
$$-f_m(\bR) = \bPsi \left(-f_m(\bSigma_p)\right) \bPsi^T,$$
which trivially proves the symmetry of $-f_m(\bR)$ and also shows that its eigenvalues are equal to $-f_m(p_i)$ for all $i=1\ldots n$. We now bound
\begin{eqnarray*}
  \abs{(-f_m(p_i))-p_i \ln \frac{1}{p_i}} &=& \abs{-f_m(p_i)+p_i \ln p_i} \\
   &=& \abs{p_i \ln p_i-f_m(p_i)}\\
   &\leq& \frac{u}{2m^2}
   \leq \epsilon\ell \ln\frac{1}{1-\ell},
\end{eqnarray*}
where the inequalities follow from Lemma~\ref{lem:prior} and our choice for $m$ from eqn.~(\ref{eqn:PD5}). This inequality holds for all $i=1\ldots n$ and implies that
$$-f_m(p_i) \geq p_i \ln \frac{1}{p_i} - \epsilon\ell \ln \frac{1}{1-\ell} \geq (1-\epsilon)\ell \ln\frac{1}{1-\ell},$$
using our upper ($1-\ell<1$) and lower ($\ell>0$) bounds on the $p_i$s. Now $\epsilon \leq 1$ proves that $-f_m(p_i)$ are non-negative for all $i=1\ldots n$ and thus $-f_m(\bR)$ is a symmetric positive semidefinite matrix; it follows that its trace is also non-negative.

We can now apply the trace estimator of Lemma~\ref{thm:trace} to get
\begin{equation}\label{eqn:PD6}
\abs{
\trace{-f_m(\bR)} - \left(-\frac1{s} \sum_{i=1}^{s} \bg_i^\top f_m(\bR) \bg_i\right)} \leq \epsilon \cdot \trace{-f_m(\bR)}.
\end{equation}
For the above bound to hold, we need to set
\begin{equation}\label{eqn:PD7}
    s=\ceil{20\ln(2/\delta)/\epsilon^2}.
\end{equation}

We now conclude as follows:
\begin{eqnarray}
\nonumber \abs{\mathcal{H}(\bR)- \aentropy{\bR}} \leq& \abs{\mathcal{H}(\bR) - \trace{-f_m(\bR)}}+ \abs{\trace{-f_m(\bR)}-\left(-\frac1{s} \sum_{i=1}^{s} \bg_i^\top f_m(\bR) \bg_i\right)}\\
   \nonumber \leq& \epsilon \mathcal{H}(\bR) + \epsilon \trace{-f_m(\bR)}\\
   \nonumber \leq& \epsilon \mathcal{H}(\bR) + \epsilon(1+\epsilon)\mathcal{H}(\bR)\\
   \nonumber \leq & 3\epsilon \mathcal{H}(\bR).
\end{eqnarray}
The first inequality follows by adding and subtracting $-\trace{f_m(\bR)}$ and using sub-additivity of the absolute value; the second inequality follows by eqns.~(\ref{eqn:PD4}) and~(\ref{eqn:PD6}); the third inequality follows again by eqn.~(\ref{eqn:PD4}); and the last inequality follows by using $\epsilon\leq 1$.

We note that the failure probability of the algorithm is at most $2\delta$ (the sum of the failure probabilities of the power method and the trace estimation algorithm).
Finally, we discuss the running time of Algorithm~\ref{alg1b}, which is equal to $\OO(s\cdot m\cdot \nnz(\bR) )$. Using the values for $m$ and $s$ from eqns.~(\ref{eqn:PD5}) and~(\ref{eqn:PD7}), the running time becomes (after accounting for the running time of Algorithm~\ref{alg:power})
$$\OO\left(\left(\sqrt{\frac{u}{\ell \ln(1/(1-\ell))}}\cdot\frac{1}{\epsilon^{2.5}}+\ln(n)\right)\ln(1/\delta)\cdot \nnz(\bR)\right).$$

\subsection{A comparison with the results of~\cite{Wihler2014}}\label{sxn:wihler}

The work of~\cite{Wihler2014} culminates in the error bounds described in Theorem 4.3 (and the ensuing discussion). In our parlance,~\cite{Wihler2014} first derives the error bound of eqn.~(\ref{eqn:PD2}).
It is worth emphasizing that the bound of eqn.~(\ref{eqn:PD2}) holds even if the $p_i$s are not necessarily strictly positive, as assumed by Theorem~\ref{thm:cheb}: the bound holds even if some of the $p_i$s are equal to zero.

Unfortunately, without imposing a lower bound assumption on the $p_i$s it is difficult to get a meaningful error bound and an efficient algorithm. Indeed, the error implied by eqn.~(\ref{eqn:PD2}) (without any assumption on the $p_i$s) necessitates setting $m$ to at least $\Omega(\sqrt{n})$ (perhaps up to a logarithmic factor, as we will discuss shortly). To understand this, note that the entropy of the density matrix $\bR$ ranges between zero and $\ln k$, where $k$ is the rank of the matrix $\bR$, i.e., the number of non-zero $p_i$'s. Clearly, $k \leq n$ and thus $\ln n$ is an upper bound for $\mathcal{H}(\bR)$. Notice that if $\mathcal{H}(\bR)$ is smaller than $n/(2m^2)$, the error bound of eqn.~(\ref{eqn:PD2}) does not even guarantee that the resulting approximation will be positive, which is, of course, meaningless as an approximation to the entropy.

In order to guarantee a relative error bound of the form $\epsilon \mathcal{H}(\bR)$ via eqn.~(\ref{eqn:PD2}), we need to set $m$ to be at least
\begin{equation}\label{eqn:PD31}
m \geq \sqrt\frac{n}{2\epsilon \mathcal{H}(\bR)},
\end{equation}
which even for ``large'' values of $\mathcal{H}(\bR)$ (i.e., values close to the upper bound $\ln n$) still implies that $m$ is $\OO(\epsilon^{-1/2}\sqrt{n/\ln n})$. Even with such a large value for $m$, we are still not done: we need an efficient trace estimation procedure for the matrix $-f_m(\bR)$. While this matrix is always symmetric, it is not necessarily positive or negative semi-definite (unless additional assumptions are imposed on the $p_i$s, like we did in Theorem~\ref{thm:cheb}). 
\section{Approaches for Hermitian Density Matrices}
Hermitian, instead of symmetric, positive definite matrices, frequently arise in quantum mechanics. The analyses of Sections~\ref{sxn:Taylor} and~\ref{sxn:cheb} focus on real density matrices; we now briefly discuss how they can be extended to Hermitian density matrices. Recall that both approaches follow the same algorithmic scheme. First, the dominant eigenvalue of the density matrix is estimated via the power method; a trace estimation follows using Gaussian trace estimators on either the truncated Taylor expansion of a suitable matrix function or on a Chebyshev polynomial approximation of the same matrix function. Interestingly, the Taylor expansions, as well as the Chebyshev polynomial approximations, both work when the input matrix is complex. However, the estimation of the dominant eigenvalue of $\bR$ poses a theoretical difficulty: to the best of our knowledge, there is no known bound for the accuracy of the power method in the case where $\bR$ is complex. Lemma~\ref{lem:power} guarantees relative error approximations to the dominant eigenvalue of real matrices, but we are not aware of any provable relative error bound for the complex case. To avoid this issue we will be using one as a (loose) upper bound for the dominant eigenvalue.

The crucial step in order to guarantee relative error approximations to the entropy of a Hermitian positive definite matrix is to guarantee relative error approximations for the trace of a Hermitian positive definite matrix. Lemma~\ref{thm:trace} assumes symmetric positive semi-definite matrices; we now prove that the same lemma can be applied on Hermitian positive definite matrices to achieve the same guarantees.

\begin{theorem}\label{thm:hermitian1}
	Every Hermitian matrix $\bA\in\comp^{n\times n}$ can be expressed as
	\begin{equation}\label{eq:hermitian1}
	\bA = \bB + i\bC,
	\end{equation}
where $\bB\in\real^{n\times n}$ is symmetric and $\bC\in\real^{n\times n}$ is anti-symmetric (or skew-symmetric).
If $\bA\in\comp^{n\times n}$ is positive semi-definite, then $\bB$ is also positive semi-definite.
\end{theorem}
\begin{proof}
The proof is trivial and uses the fact that for any Hermitian (symmetric) positive semi-definite matrix all eigenvalues are real and greater than zero.
\end{proof}

\begin{theorem}\label{thm:hermitian2}
	The trace of a Hermitian matrix $\bA\in\comp^{n\times n}$ expressed as in eqn.~\eqref{eq:hermitian1} is equal to the trace of its real part:
	\begin{equation*}\label{eq:hermitian2}
	\trace{\bA} = \trace{\bB}.
	\end{equation*}
\end{theorem}
\begin{proof}
Using $\trace{\bA} = \trace{\bA^\top}$, it is easy to see that
\begin{equation*}
\trace{\bA} =  \trace{\bB+i\bC} =  \trace{\bB} + i\trace{\bC}=\trace{\bB^\top}+i\trace{\bC^\top}
= \trace{\bB}.
\end{equation*}
The last equality follows by noticing that the only way for the equality to hold for a skew-symmetric matrix $\bC$ is if $\trace{\bC^\top} = -\trace{\bC^\top}$. This is true only if $\bC$ is the all-zeros matrix.
\end{proof}

In words, Theorem~\ref{thm:hermitian2} states that the trace of a Hermitian matrix equals the trace of its real part. Similarly, Theorem~\ref{thm:hermitian1} states that the real part of a Hermitian positive semi-definite matrix is symmetric positive semi-definite. Combining both theorems we conclude that we can estimate the trace of a Hermitian positive definite matrix up to relative error, using the Gaussian trace estimator of Lemma~\ref{thm:trace} on its real part. Therefore, both approaches generalize to Hermitian positive definite matrices using one as an upper bound instead of $u$ for the dominant eigenvalue. Algorithms~\ref{alg1a_complex} and~\ref{alg1b_complex} are modified versions of Algorithms~\ref{alg1a} and~\ref{alg1b} respectively that work on Hermitian inputs (the function ${\rm Re}(\cdot)$ returns the real part of its argument in an entry-wise manner).
	\begin{algorithm}[H]
		\begin{algorithmic}[1]
			\STATE {\bf{INPUT}}: $\bR \in \comp^{n \times n}$, accuracy parameter $\varepsilon > 0$, failure probability $\delta$, and integer $m >0$.
			\STATE Set $s = \left\lceil 20 \ln(2/\delta) / \varepsilon^2 \right \rceil$.
			\STATE Let $\bg_1,\bg_2,\ldots, \bg_s \in \real^n$ be i.i.d. random Gaussian vectors.
			\STATE \textbf{OUTPUT}: return
			$$\aentropy{\bR} = \frac{1}{s}\sum_{i=1}^{s}\sum_{k=1}^{m}\frac{\bg_i^\top \left({\rm Re}\left[\bR (\bI_n-\bR)^k\right]\right) \bg_i}{k}.$$
		\end{algorithmic}
		\caption{\small{A Taylor series approach to estimate the entropy.}}
		\label{alg1a_complex}
	\end{algorithm}
\begin{algorithm}[h!]
	\begin{algorithmic}[1]
		\STATE {\bf{INPUT}}: $\bR \in \mathbb{C}^{n \times n}$, accuracy parameter $\varepsilon > 0$, failure probability $\delta$, and integer $m >0$.
		\STATE Set $s = \left\lceil 20 \ln(2/\delta) / \varepsilon^2 \right \rceil$.
		\STATE Let $\bg_1,\bg_2,\ldots, \bg_s \in \mathbb{R}^n$ be i.i.d. random Gaussian vectors.
		\STATE{\bf{OUTPUT}:
			$\aentropy{\bR} = -\frac1{s}\sum_{i=1}^{s} \bg_i^\top \left({\rm Re}\left[f_m(\bR)\right]\right) \bg_i.$
		}
	\end{algorithmic}
	\caption{{A Chebyschev polynomial-based approach to estimate the entropy.}}
	\label{alg1b_complex}
\end{algorithm}
\noindent Theorems~\ref{thm:taylor_complex} and~\ref{thm:cheb_complex} are our main quality-of-approximation results for Algorithm~\ref{alg1a_complex} and~\ref{alg1b_complex}.
	\begin{theorem}
		\label{thm:taylor_complex}
		Let $\bR$ be a complex density matrix such that all probabilities $p_i$, $i=1\ldots n$ satisfy $0 < \ell \leq p_i$. Let $\aentropy{\bR}$ be the output of Algorithm~\ref{alg1a_complex} on inputs $\bR$, $m$, and $\epsilon<1$.
		Then, with probability at least $1-\delta$,
		$$
		\abs{ \aentropy{\bR} - \entropy{\bR} }  \leq
		2\epsilon \entropy{\bR},$$
		by setting  $m=\left\lceil\frac{1}{\ell}\ln{\frac{1}{\epsilon}}\right\rceil$.
		The algorithm runs in time
		$$\OO\left(\frac{\ln(1/\epsilon)}{\ell\cdot\epsilon^2}\cdot \ln(1/\delta)\cdot nnz(\bR)\right).$$ %
	\end{theorem}
\begin{theorem}
	\label{thm:cheb_complex}
	Let $\bR$ be a density matrix such that all probabilities $p_i$, $i=1\ldots n$ satisfy $0 < \ell \leq p_i$. Let $\aentropy{\bR}$ be the output of Algorithm~\ref{alg1b_complex} on inputs $\bR$, $m$, and $\epsilon<1$.
	Then, with probability at least $1-\delta$,
	$$
	\abs{ \aentropy{\bR} - \entropy{\bR} }  \leq
	3\epsilon \entropy{\bR},$$
	by setting  $m = \sqrt{\frac{1}{2\epsilon\ell \ln(1/(1-\ell))}}$. The algorithm runs in time
	$$\OO\left(\sqrt{\frac{1}{\ell \ln(1/(1-\ell))}}\cdot\frac{1}{\epsilon^{2.5}}\ln(1/\delta)\cdot \nnz(\bR)\right).$$
\end{theorem}

\section{An approach via random projection matrices}\label{sxn:rp}

Finally, we focus on perhaps the most interesting special case: the setting where at most $k$ (out of $n$, with $k \ll n$) of the probabilities $p_i$ of the density matrix $\bR$ of eqn.~(\ref{eqn:R}) are non-zero. In this setting, we prove that elegant random-projection-based techniques achieve relative error approximations to all probabilities $p_i$, $i=1\ldots k$. The running time of the proposed approach depends on the particular random projection that is used and can be made to depend on the sparsity of the input matrix.

\subsection{Algorithm and Main Theorem}

The proposed algorithm uses a random projection matrix $\bPi$ to create a ``sketch'' of $\bR$ in order to approximate the $p_i$s.
\begin{algorithm}
\centerline{
\caption{Approximating the entropy via random projection matrices}\label{alg:EN}
}
\begin{algorithmic}[1]
\STATE {\bf{INPUT}}: Integer $n$ (dimensions of matrix $\bR$) and integer $k$ (with rank of $\bR$ at most $k \ll n$, see eqn.~(\ref{eqn:R})).
\STATE Construct the random projection matrix $\bPi \in \mathbb{R}^{n \times s}$ (see Section~\ref{sxn:constructpi} for details on $\bPi$ and $s$).
\STATE Compute $\tilde{\bR} = \bR\bPi \in \mathbb{R}^{n \times s}$.
\STATE Compute and return the (at most) $k$ non-zero singular values of $\tilde{\bR}$, denoted by $\tilde{p}_i$, $i=1\ldots k$.
\STATE {\bf{OUTPUT}}: $\tilde{p}_i$, $i=1\ldots k$ and $\aentropy{\bR} = \sum_{i=1}^k \tilde{p}_i \ln \frac{1}{\tilde{p}_i}$.
\end{algorithmic}
\end{algorithm}
In words, Algorithm~\ref{alg:EN} creates a sketch of the input matrix $\bR$ by post-multiplying $\bR$ by a random projection matrix; this is a well-known approach from the RandNLA literature (see~\cite{Drineas2016} for details). Assuming that $\bR$ has rank at most $k$, which is equivalent to assuming that at most $k$ of the probabilities $p_i$ in eqn.~(\ref{eqn:R}) are non-zero (e.g., the system underlying the density matrix $\bR$ has at most $k$ pure states), then the rank of $\bR\bPi$ is also at most $k$. In this setting, Algorithm~\ref{alg:EN} returns the non-zero singular values of $\bR\bPi$ as approximations to the $p_i$, $i=1\ldots k$.

The following theorem is our main quality-of-approximation result for Algorithm~\ref{alg:EN}.
\begin{theorem}
	\label{thm:rp}
Let $\bR$ be a density matrix with at most $k \ll n$ non-zero probabilities and let $\epsilon<1/2$ be an accuracy parameter. Then, with probability at least $0.9$, the output of Algorithm~\ref{alg:EN} satisfies
$$\abs{p_i^2 - \tilde{p}_i^2} \leq \epsilon p_i^2$$
for all $i=1\ldots k$. Additionally,
$$\abs{\mathcal{H}(\bR)-\aentropy{\bR}}\leq \sqrt{\epsilon}\mathcal{H}(\bR)+\sqrt{\frac{3}{2}}\epsilon.$$
Algorithm~\ref{alg:EN} (combined with Algorithm~\ref{alg:IS} below) runs in time
$$\OO\left(\nnz(\bR)+nk^4/\epsilon^4\right).$$
\end{theorem}
\noindent Comparing the above result with Theorems~\ref{thm:taylor} and~\ref{thm:cheb}, we note that the above theorem does not necessitate imposing any constraints on the probabilities $p_i$, $i=1\ldots k$. Instead, it suffices to have $k$ non-zero probabilities. The final result is an additive-relative error approximation to the entropy of $\bR$ (as opposed to the relative error approximations of  Theorems~\ref{thm:taylor} and~\ref{thm:cheb}); under the mild assumption $\mathcal{H}(\bR) \geq \sqrt{\epsilon}$, the above bound becomes a true relative error approximation\footnote{Recall that $\mathcal{H}(\bR)$ ranges between zero and $\ln k$.}.

\subsection{Two constructions for the random projection matrix}\label{sxn:constructpi}

We now discuss two constructions for the matrix $\bPi$ and we cite two bounds regarding these constructions from prior work that will be useful in our analysis. The first construction is the subsampled Hadamard Transform, a simplification of the Fast Johnson-Lindenstrauss Transform of~\cite{Ailon2009}; see~\cite{Drineas2011,Tropp2010} for details. We do note that even though it appears that Algorithm~\ref{alg:IS} is always better than Algorithm~\ref{alg:SHD} (at least in terms of their respective \textit{theoretical} running times), both algorithms are worth evaluating experimentally: in particular, prior work~\cite{Paul2013} has reported that Algorithm~\ref{alg:SHD} often outperforms Algorithm~\ref{alg:IS} in terms of empirical accuracy and running time when the input matrix is dense, as is often the case in our setting. Therefore, we choose to present results (theoretical and empirical) for both well-known constructions of $\bPi$ (Algorithms~\ref{alg:SHD} and~\ref{alg:IS}).
\begin{algorithm}
\centerline{
\caption{The subsampled Randomized Hadamard Transform}\label{alg:SHD}
}
\begin{algorithmic}[1]
\STATE \textbf{INPUT:} integers $n, s>0$ with $s \ll n$.
\STATE Let $\bS$ be an empty matrix.
\STATE \textbf{For} $t=1,\ldots,s$ (i.i.d. trials with replacement) \textbf{select uniformly at random} an integer from $\left\{1,2,\ldots,n\right\}$.
\STATE \textbf{If} $i$ is selected, \textbf{then} append the column vector $\be_i$ to $\bS$, where $\be_i \in \mathbb{R}^n$ is the $i$-th canonical vector.
\STATE Let $\bH \in \mathbb{R}^{n\times n}$ be the normalized Hadamard transform matrix.
\STATE Let $\bD \in \mathbb{R}^{n \times n}$ be a diagonal matrix with
$$
\bD_{ii}
            = \left\{ \begin{array}{ll}
                         +1 & \mbox{, with probability $1/2$} \\
                         -1 & \mbox{, with probability $1/2$} \\
                      \end{array}
              \right.
$$
\STATE \textbf{OUTPUT:} $\bPi =\bD\bH\bS \in \mathbb{R}^{n \times s}$.
\end{algorithmic}
\end{algorithm}

\noindent The following result has appeared in~\cite{Drineas2011,Tropp2010,Woodruff2014}.
\begin{lemma}
	\label{lem:rp1}
Let $\bU \in \mathbb{R}^{n \times k}$ such that $\bU^T\bU=\bI_k$ and let $\bPi \in \mathbb{R}^{n \times s}$ be constructed by Algorithm~\ref{alg:SHD}. Then, with probability at least 0.9,
$$\left\|\frac{n}{k}\bU^T\bPi\bPi^T\bU - \bI_k\right\|_2 \leq \epsilon,$$
by setting $s = \OO\left(\left(k + \log n\right)\cdot \frac{\log k}{\epsilon^2}\right)$.
\end{lemma}

\noindent Our second construction is the input sparsity transform of~\cite{Clarkson2013}. This major breakthrough was further analyzed in~\cite{Meng2013,Nelson2013} and we present the following result from~\cite[Appendix A1]{Meng2013}.

\begin{algorithm}
\centerline{
\caption{An input-sparsity transform}\label{alg:IS}
}
\begin{algorithmic}[1]
\STATE \textbf{INPUT:} integers $n, s>0$ with $s \ll n$.
\STATE Let $\bS$ be an empty matrix.
\STATE \textbf{For} $t=1,\ldots,n$ (i.i.d. trials with replacement) \textbf{select uniformly at random} an integer from $\left\{1,2,\ldots,s\right\}$.
\STATE \textbf{If} $i$ is selected, \textbf{then} append the row vector $\be_i^T$ to $\bS$, where $\be_i \in \mathbb{R}^s$ is the $i$-th canonical vector.
\STATE Let $\bD \in \mathbb{R}^{n \times n}$ be a diagonal matrix with
$$
\bD_{ii}
            = \left\{ \begin{array}{ll}
                         +1 & \mbox{, with probability $1/2$} \\
                         -1 & \mbox{, with probability $1/2$} \\
                      \end{array}
              \right.
$$
\STATE \textbf{OUTPUT:} $\bPi = \bD\bS \in \mathbb{R}^{n \times s}$.
\end{algorithmic}
\end{algorithm}
\begin{lemma}
	\label{lem:rp2}
Let $\bU \in \mathbb{R}^{n \times k}$ such that $\bU^T\bU=\bI_k$ and let $\bPi \in \mathbb{R}^{n \times s}$ be constructed by Algorithm~\ref{alg:IS}. Then, with probability at least 0.9,
$$\|\bU^T\bPi\bPi^T\bU - \bI_k\|_2 \leq \epsilon,$$
by setting $s = \OO(k^2/\epsilon^2)$.
\end{lemma}
\noindent We refer the interested reader to~\cite{Nelson2013} for improved analyses of Algorithm~\ref{alg:IS} and its variants.

\subsection{Proof of Theorem~\ref{thm:rp}}

At the heart of the proof of Theorem~\ref{thm:rp} lies the following perturbation bound from~\cite{Demmel1992} (Theorem 2.3).
\begin{theorem}
	\label{thm:DV}
Let $\bD \bA \bD$ be a symmetric positive definite matrix such that $\bD$ is a diagonal matrix and $\bA_{ii}=1$ for all $i$. Let $\bD\bE \bD$ be a perturbation matrix  such that $\TNorm{\bE}< \lambda_{\min}(\bA)$. Let $\lambda_i$ be the $i$-the eigenvalue of $\bD \bA \bD$ and let $\lambda_i'$ be the $i$-th eigenvalue of $\bD (\bA+\bE) \bD$. Then, for all $i$,
$$ \abs{\lambda_i-\lambda_i'} \leq \frac{\TNorm{\bE}}{\lambda_{\min}(\bA)}.$$
\end{theorem}
\noindent We note that ${\lambda_{\min}(\bA)}$ in the above theorem is a real, strictly positive number\footnote{This follows from the fact that $\bA$ is a symmetric positive definite matrix and the inequality $0\leq\TNorm{\bE}< \lambda_{\min}(\bA)$.}. Now consider the matrix $\bR \bPi \bPi^T \bR^T$; we will use the above theorem to argue that its singular values are good approximations to the singular values of the matrix $\bR \bR^T$. Recall that $\bR = \bPsi \bSigma_p \bPsi^T$ where $\bPsi$ has orthonormal columns. Note that the eigenvalues of $\bR \bR^T = \bPsi \bSigma_p^2 \bPsi^T$ are equal to the eigenvalues of the matrix $\bSigma_p^2$; similarly, the eigenvalues of
$\bPsi \bSigma_p \bPsi^T \bPi \bPi^T \bPsi \bSigma_p \bPsi^T$ are equal to the eigenvalues of $\bSigma_p \bPsi^T \bPi \bPi^T \bPsi \bSigma_p$. Thus, we can compare the matrices
$$\bSigma_p \bI_k \bSigma_p \quad \mbox{and} \quad \bSigma_p \bPsi^T \bPi \bPi^T \bPsi \bSigma_p.$$
In the parlance of Theorem~\ref{thm:DV}, $\bE=\bPsi^T \bPi \bPi^T \bPsi-\bI_k$. Applying either Lemma~\ref{lem:rp1} (after rescaling the matrix $\bPi$) or Lemma~\ref{lem:rp2}, we immediately get that $\TNorm{\bE_{A}}\leq \epsilon <1$ with probability at least 0.9. Since $\lambda_{\min}(\bI_k)=1$, the assumption of Theorem~\ref{thm:DV} is satisfied. We note that the eigenvalues of $\bSigma_p \bI_k \bSigma_p$ are equal to $p_i^2$ for $i=1\ldots k$ (all positive, which guarantees that the matrix $\bSigma_p \bI_k \bSigma_p$ is symmetric positive definite, as mandated by Theorem~\ref{thm:DV}) and the eigenvalues of $\bSigma_p \bPsi^T \bPi \bPi^T \bPsi \bSigma_p$ are equal to $\tilde{p}_i^2$, where $\tilde{p}_i$ are the singular values of $\bSigma_p \bPsi^T \bPi$. (Note that these are exactly equal to the outputs returned by Algorithm~\ref{alg:EN}, since the singular values of $\bSigma_p \bPsi^T \bPi$ are equal to the singular values of $\bPsi\bSigma_p \bPsi^T \bPi=\bR\bPi$). Thus, we can conclude:
\begin{equation}\label{eqn:PD10}
  \abs{p_i^2 - \tilde{p}_i^2} \leq \epsilon p_i^2.
\end{equation}
The above result guarantees that all $p_i$s can be approximated up to \textit{relative error} using Algorithm~\ref{alg:EN}. We now investigate the implication of the above bound to approximating the von Neumann entropy of $\bR$. Indeed,
\begin{eqnarray*}
\sum_{i=1}^k \tilde{p}_i \ln \frac{1}{\tilde{p}_i} &\leq&  \sum_{i=1}^k (1+\epsilon)^{1/2} p_i \ln \frac{1}{(1-\epsilon)^{1/2}p_i}\\
&\leq& {(1+\epsilon)^{1/2}\left(\sum_{i=1}^k p_i \ln \frac{1}{p_i}+\sum_{i=1}^k p_i \ln \frac{1}{(1-\epsilon)^{1/2}}\right)}\\
&=& (1+\epsilon)^{1/2}\mathcal{H}(\bR)+\frac{\sqrt{1+\epsilon}}{2}\ln \frac{1}{1-\epsilon}\\
&\leq& (1+\epsilon)^{1/2}\mathcal{H}(\bR)+\frac{\sqrt{1+\epsilon}}{2}\ln (1+2\epsilon)\\
&\leq& (1+\sqrt{\epsilon})\mathcal{H}(\bR)+\sqrt{\frac 3 2}\epsilon.
\end{eqnarray*}
In the second to last inequality we used $1/(1-\epsilon) \leq 1+2\epsilon$ for any $\epsilon \leq 1/2$ and in the last inequality we used $\ln(1+2\epsilon)\leq 2\epsilon$ for $\epsilon\in (0,1/2)$. Similarly, we can prove that:
$$\sum_{i=1}^k \tilde{p}_i \ln \frac{1}{\tilde{p}_i} \geq (1-\sqrt{\epsilon})\mathcal{H}(\bR)-\frac{1}{2}\epsilon.$$
Combining, we get
$$\abs{\sum_{i=1}^k \tilde{p}_i \ln \frac{1}{\tilde{p}_i} - \mathcal{H}(\bR)}\leq \sqrt{\epsilon}\mathcal{H}(\bR)+\sqrt{\frac{3}{2}}\epsilon.$$
We conclude by discussing the running time of Algorithm~\ref{alg:EN}. Theoretically, the best choice is to combine the matrix $\bPi$ from Algorithm~\ref{alg:IS} with Algorithm~\ref{alg:EN}, which results in a running time
$$\OO\left(\nnz(\bR)+nk^4/\epsilon^4\right).$$
\subsection{The Hermitian case}
The above approach via random projections critically depends on Lemmas~\ref{lem:rp1} and~\ref{lem:rp2}, which, to the best of our knowledge, have only been proven for the real case. These results are typically proven using matrix concentration inequalities, which are well-explored for sums of random real matrices but less explored for sums of real complex matrices. We leave it as an open problem to extend the theoretical analysis of our approach to the Hermitian case.

\section{Experiments}\label{sxn:experiments}

In this section we report experimental results in order to demonstrate the practical efficiency of our algorithms. We show that our algorithms are \textit{both} numerically accurate \textit{and} computationally efficient. Our algorithms were implemented in Matlab R2016a on a compute node with two 10-Core Intel Xeon-E5 processors (2.60GHz) and 512 GBs of RAM.

We generated random density matrices for most of which we used the QETLAB Matlab toolbox~\cite{qetlab} to derive (real-valued) density matrices of size $5,000 \times 5,000$, on which most of our extensive evaluations were run. We also tested our methods on a much larger $30,000 \times 30,000$ density matrix, which was close to the largest matrix that Matlab would allow us to load. We used the function~\texttt{RandomDensityMatrix} of QETLAB and the Haar measure; we also experimented with the Bures measure to generate random matrices, but we did not observe any qualitative differences worth reporting. Recall that exactly computing the Von-Neumann entropy using eqn.~\eqref{eqn:earlyR} presumes knowledge of the entire spectrum of the matrix; to compute all singular values of a matrix we used the~\texttt{svd} function of Matlab. The accuracy of our proposed approximation algorithms was evaluated by measuring the relative error; wall-clock times were reported in order to quantify the speedup that our approximation algorithms were able to achieve.

\subsection{Empirical results for the Taylor and Chebyshev approximation algorithms}

We start by reporting results on the Taylor and Chebyshev approximation algorithms, which have two sources of error: the number of terms that are retained in either the Taylor series expansion or the Chebyshev polynomial approximation \textit{and} the trace estimation that is used in both approximation algorithms. We will separately evaluate the accuracy loss that is contributed by each source of error in order to understand the behavior of the proposed approximation algorithms.

Consider a $5,000 \times 5,000$ random density matrix and let $m$ (the number of terms retained in the Taylor series approximation or the degree of the polynomial used in the Chebyshev polynomial approximation) range between five and 30 in increments of five. Let $s$, the number of random Gaussian vectors used to estimate the trace, be set to $\{50, 100, 200, 300\}$. Recall that our error bounds for Algorithms~\ref{alg1a} and~\ref{alg1b} depend on $u$, an estimate for the largest eigenvalue of the density matrix. We used the power method to estimate the largest eigenvalue (let $\tilde{\lambda}_{\max}$ be the estimate) and we set $u$ to $\tilde{\lambda}_{\max}$ and $6\tilde{\lambda}_{\max}$. Figures~\ref{fig:RE5K_smax} and~\ref{fig:RE5K_6smax} show the relative error (out of 100\%) for all combinations of $m$, $s$, and $u$ for the Taylor and Chebyshev approximation algorithms. It is worth noting that we also report the error when no trace estimation (NTE) is used in order to highlight that most of the accuracy loss is due to the Taylor/Chebyshev approximation and not the trace estimation.

We observe that the relative error is always small, typically close to $1$-$2\%$, for any choice of the parameters $s$, $m$, and $u$. The Chebyshev algorithm returns better approximations when $u$ is an overestimate for $\lambda_{\max}$ while the two algorithms are comparable (in terms of accuracy) where $u$ is very close to $\lambda_{\max}$, which agrees with our theoretical results. We also note that estimating the largest eigenvalue incurs minimal computational cost (less than one second). The NTE line (no trace estimation) in the plots serves as a lower bound for the relative error. Finally, we note that computing the exact Von-Neumann entropy took approximately $1.5$ minutes for matrices of this size.

\begin{figure}[h!]
	\centering
	\includegraphics[width=0.7\textwidth]{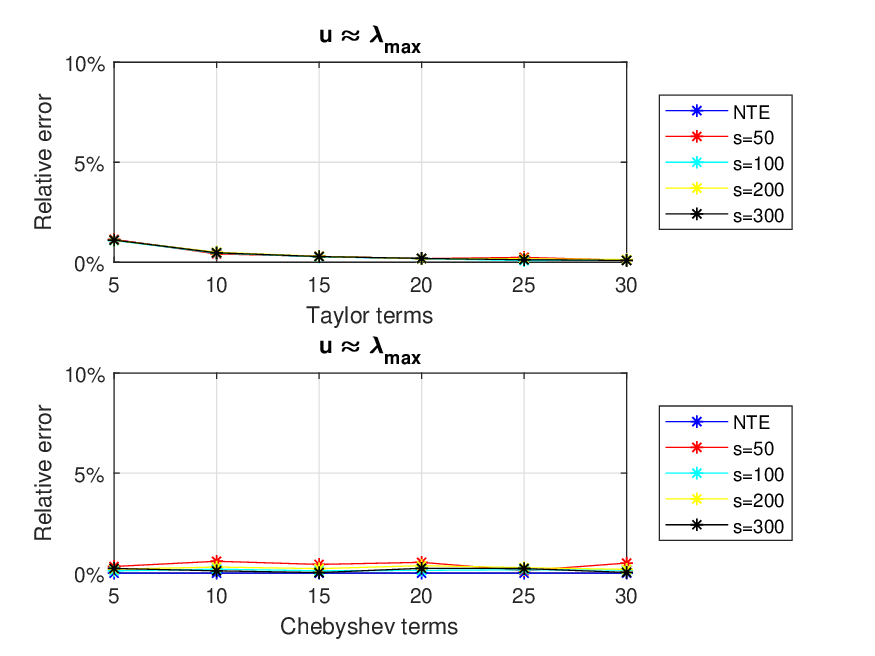}
	\caption{Relative error for $5,000 \times 5,000$ density matrix using the Taylor and the Chebyshev approximation algorithms with $u=\tilde{\lambda}_{\max}$.}
	\label{fig:RE5K_smax}	
\end{figure}

\begin{figure}[h!]
	\centering
	\includegraphics[width=0.7\textwidth]{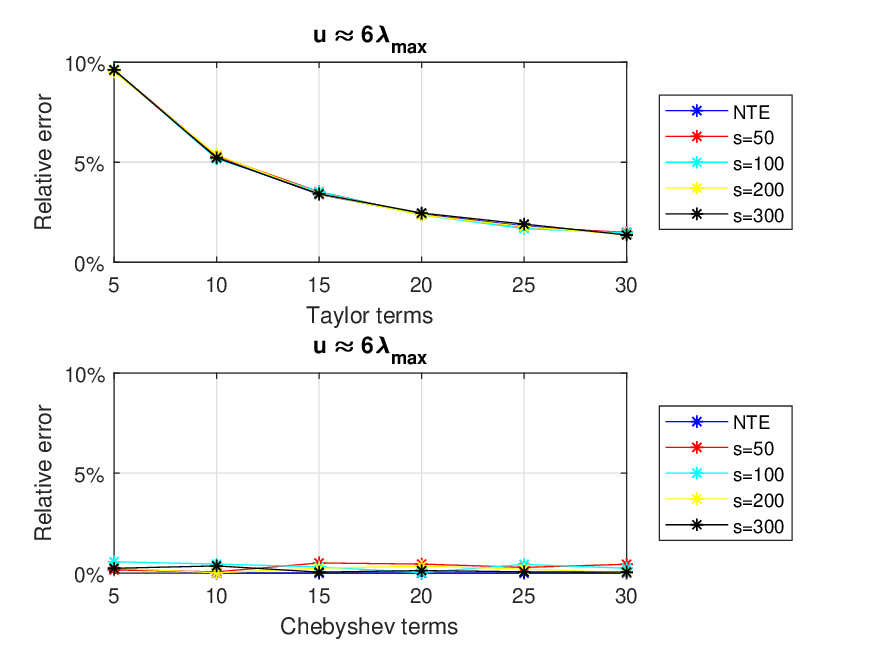}
	\caption{Relative error for $5,000 \times 5,000$ density matrix using the Taylor and the Chebyshev approximation algorithms with $u=6\tilde{\lambda}_{\max}$.}
	\label{fig:RE5K_6smax}
\end{figure}

\begin{figure}[h!]
	\centering
	\includegraphics[width=0.7\textwidth]{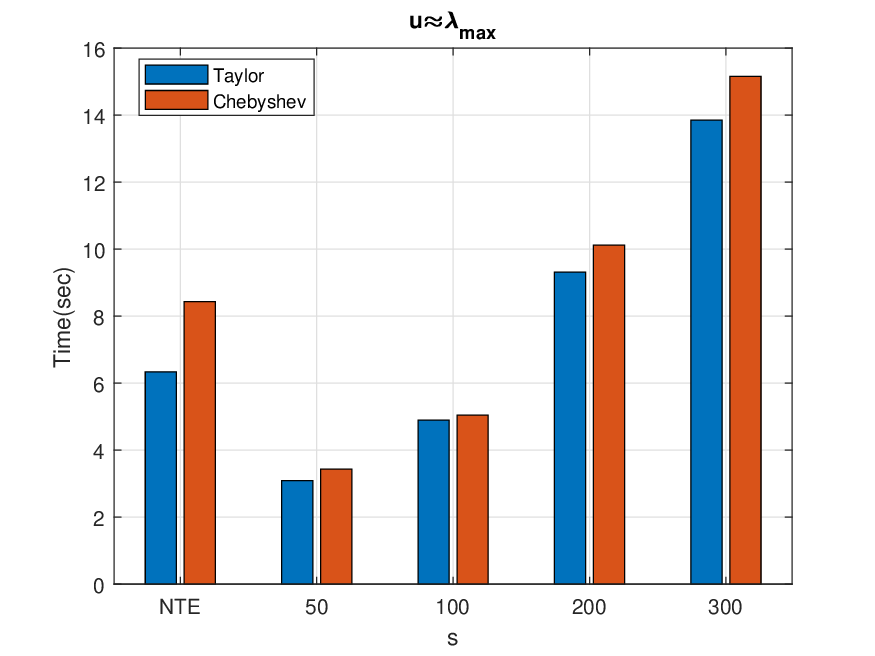}
	\caption{Time (in seconds) to run the approximate algorithms for the $5,000 \times 5,000$ density matrix for $m=5$. \textit{Exactly} computing the Von-Neumann entropy took approximately 90 seconds.}\label{fig:time5K}
\end{figure}

The second dataset that we experimented with was a much larger density matrix of size $30,000 \times 30,000$. This matrix was the largest matrix for which the memory was sufficient to perform operations like the full SVD. Notice that since the increase in the matrix size is six-fold compared to the previous one and SVD's running time grows cubically with the input size, we expect the running time to compute the exact SVD to be roughly $6^3 \cdot 90$ seconds, which is approximately 5.4 hours; indeed, the exact computation of the Von-Neumann entropy took approximately $5.6$ hours. We evaluated both the Taylor and the Chebyshev approximation schemes by setting the parameters $m$ and $s$ to take values in the sets $\{5,10,15,20\}$ and $\{50, 100, 200\}$, respectively. The parameter $u$ was set to $\tilde{\lambda}_{\max}$, where the latter value was computed using the power method, which took approximately $3.6$ minutes. We report the wall-clock running times and relative error (out of 100\%) in Figures~\ref{fig:Time30Ksmax} and~\ref{fig:RE30Ksmax}.

We observe that the relative error is always less than $1\%$ for both methods, with the Chebyshev approximation yielding almost always slightly better results. Note that our Chebyshev-polynomial-based approximation algorithm significantly outperformed the exact computation: e.g., for $m=5$ and $s = 50$, our estimate was computed in less than ten minutes and achieved less than $.2\%$ relative error.

\begin{figure}[!h]
	\centering	\includegraphics[width=0.7\textwidth]{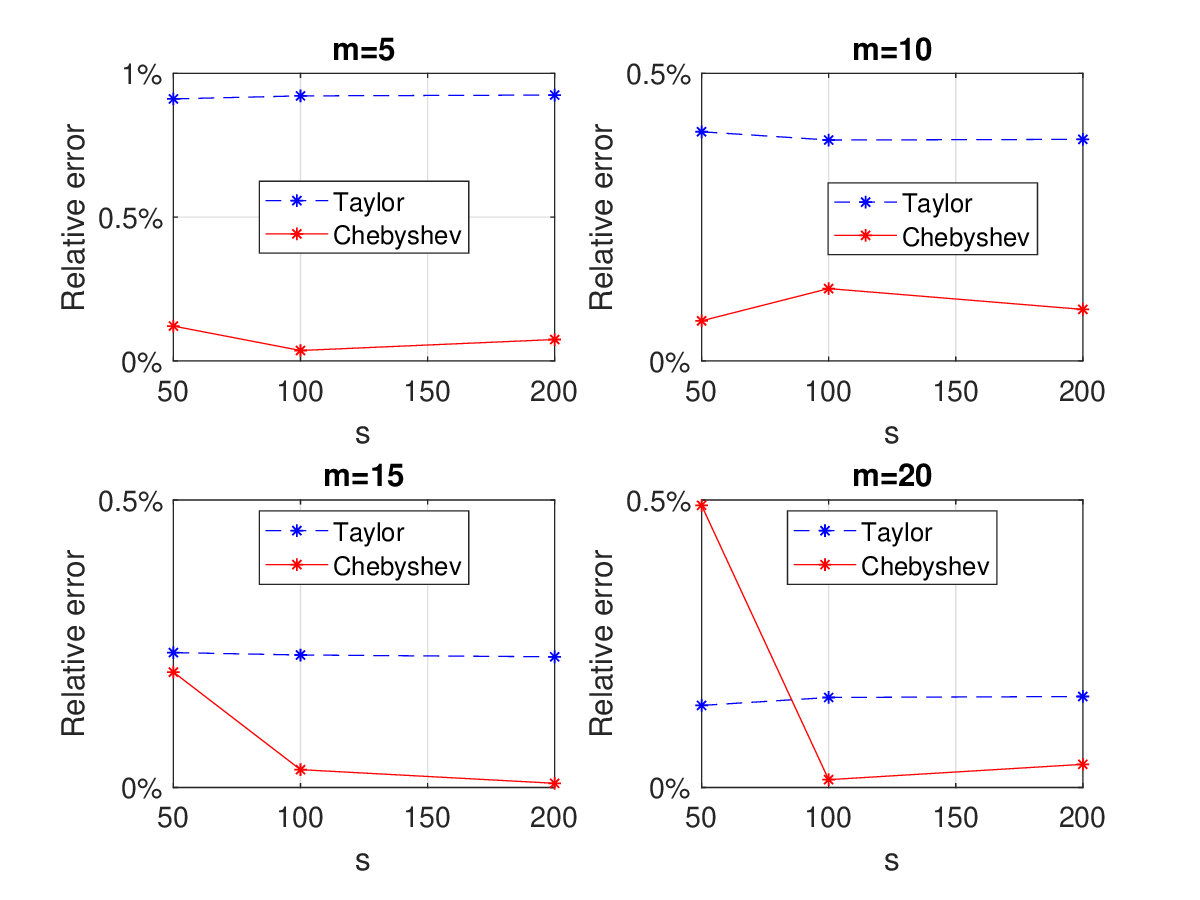}
	\caption{Relative error for $30,000 \times 30,000$ density matrix using the Taylor and the Chebyshev approximation algorithms with $u=\tilde{\lambda}_{\max}$.}
	\label{fig:RE30Ksmax}
\end{figure}

\begin{figure}[!h]
	\centering
	\includegraphics[width=0.7\textwidth]{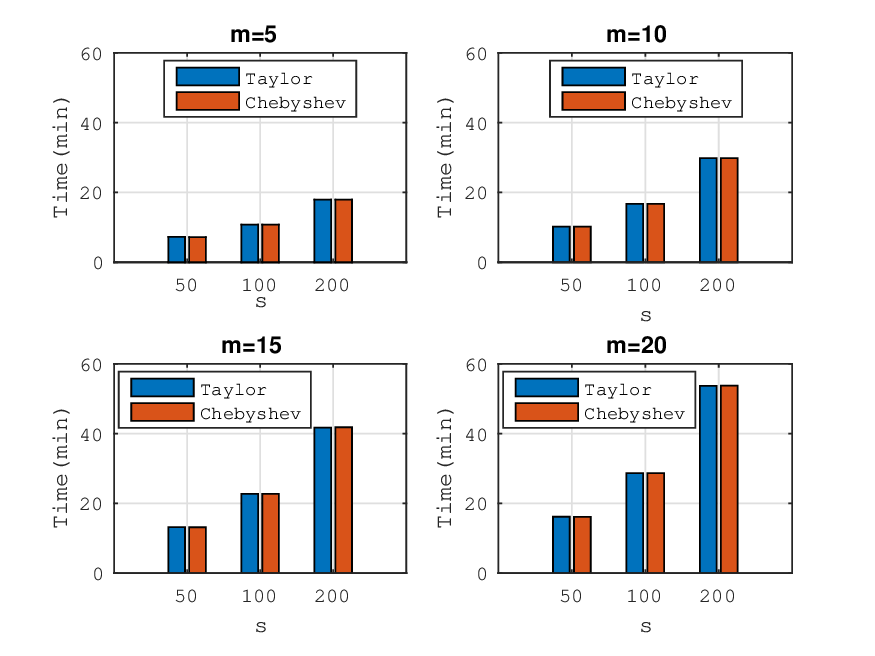}
	\caption{Wall-clock times: Taylor approximation (blue) and Chebyshev approximation (red) for $u=\tilde{\lambda}_{\max}$. Exact computation needed approximately 5.6 hours.}
	\label{fig:Time30Ksmax}
\end{figure}

The third dataset we experimented with was the tridiagonal matrix from \cite[Section 5.1]{Han2015}:
\begin{equation}
\bA = \begin{bmatrix}
2 & -1 & 0 & \dots & 0\\
-1&  2 & -1 & \ddots & \vdots\\
0 & \ddots & \ddots & \ddots & 0\\
\vdots &\ddots & -1 & 2 & -1 \\
0 & \dots & 0 & -1 & 2 \\
\end{bmatrix}\label{eq:poisson_matrix}
\end{equation}
This matrix is the coefficient matrix of the discretized one-dimensional Poisson equation:
$$ f(x) = -\frac{d^2v_x}{dx}$$
defined in the interval $[0,1]$ with Dirichlet boundary conditions $v(0) = v(1) = 0$. We normalize $\bA$ by dividing it with its trace in order to make it a density matrix.
Consider the $5,000 \times 5,000$ normalized matrix $\bA$ and let $m$ (the number of terms retained in the Taylor series approximation or the degree of the polynomial used in the Chebyshev polynomial approximation) range between five and 30 in increments of five. Let $s$, the number of random Gaussian vectors used for estimating the trace be set to 50, 100, 200, or 300. We used the formula
\begin{equation}\lambda_i = \frac{4}{2n}\sin^2\left(\frac{i\pi}{2n + 2}\right),\quad i=1,\dots,n\label{eq:trig_eigs}\end{equation}
to compute the eigenvalues of $\bA$ (after normalization) and we set $u$ to ${\lambda}_{\max}$ and $6{\lambda}_{\max}$.
Figures~\ref{fig:RE5K_banded_smax} and~\ref{fig:RE5K_banded_6smax} show the relative error (out of 100\%) for all combinations of $m$, $s$, and $u$ for the Taylor and Chebyshev approximation algorithms. We also report the error when no trace estimation (NTE) is used.

We observe that the relative error is higher than the one observed for the $5,000 \times 5,000$ random density matrix. We report wall-clock running times in Figure~\ref{fig:time5K_banded}. The Chebyshev-polynomial-based algorithm returns better approximations for all choices of the parameters and, in most cases, is faster than the Taylor-polynomial-based algorithm, e.g. for $m=5$, $s=50$ and $u={\lambda}_{\max}$, our estimate was computed in about two seconds and achieved less than $.5\%$ relative error.

\begin{figure}[!h]
	\centering
	\includegraphics[width=0.7\textwidth]{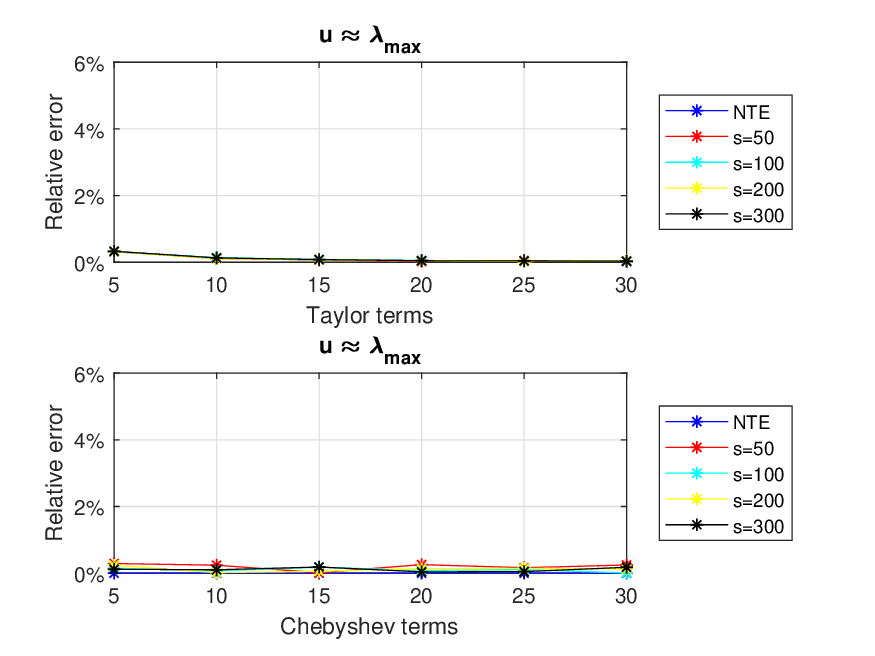}
	\caption{Relative error for $5,000 \times 5,000$ tridiagonal density matrix using the Taylor and the Chebyshev approximation algorithms with $u={\lambda}_{\max}$.}
	\label{fig:RE5K_banded_smax}	
\end{figure}

\begin{figure}[h!]
	\centering
	\includegraphics[width=0.7\textwidth]{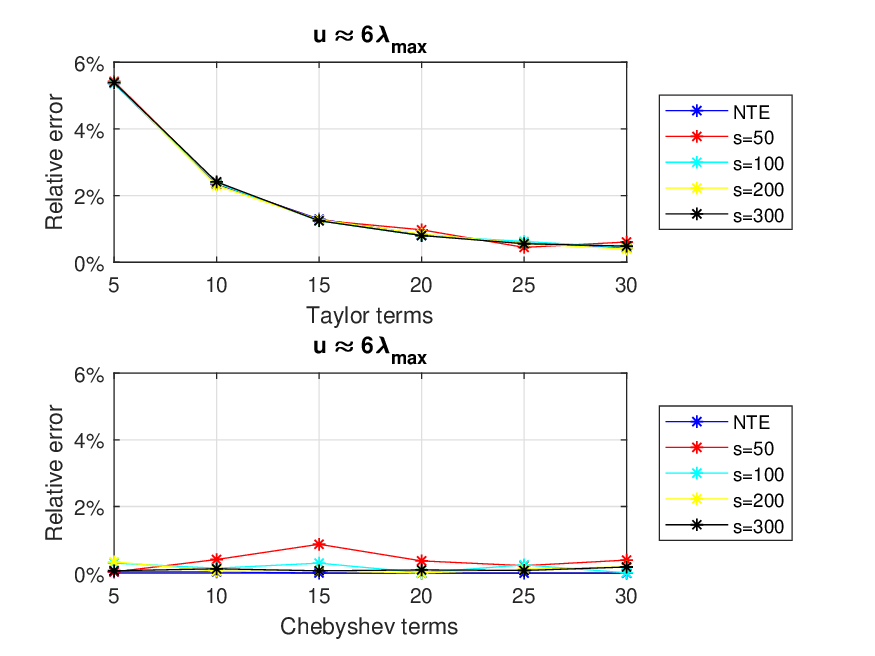}
	\caption{Relative error for $5,000 \times 5,000$ tridiagonal density matrix using the Taylor and the Chebyshev approximation algorithms with $u=6{\lambda}_{\max}$.}
	\label{fig:RE5K_banded_6smax}
\end{figure}

\begin{figure}[h!]
	\centering
	\includegraphics[width=0.7\textwidth]{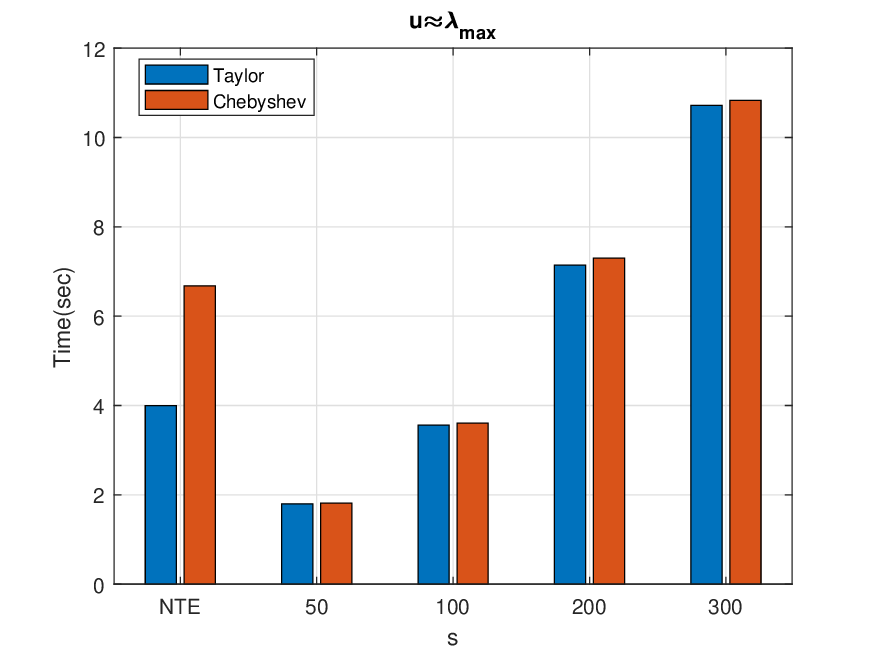}
	\caption{Wall-clock times: Taylor approximation (blue) and Chebyshev approximation (red) for $m=5$. Exact computation needed approximately $30$ seconds. }\label{fig:time5K_banded}
\end{figure}

We further considered a $10^{8}\times 10^8$ tridiagonal matrix of the form of eqn.~\eqref{eq:poisson_matrix}. Although an exact computation of the singular values of $\bA$ is not feasible (at least with our computational resources), such a computation is not necessary since eqn.~\eqref{eq:trig_eigs} provides a closed formula for its eigenvalues and, thus, its entropy. Let $m$ (the number of terms retained in the Taylor series approximation or the degree of the polynomial used in the Chebyshev polynomial approximation) be equal to five or ten and let $s$, the number of random Gaussian vectors used to estimate the trace be equal to 50 or 100. Figures~\ref{fig:RE10B_banded_smax} and~\ref{fig:Time10B_banded_smax} show the relative error (out of 100\%) and the runtime, respectively, for all combinations of $m$ and $s$ for both the Taylor and Chebyshev approximation algorithms. We observe that in both cases we estimated the entropy in less than ten minutes with a relative error below $0.15\%$.

\begin{figure}[h!]
	\centering
	\includegraphics[width=0.7\textwidth]{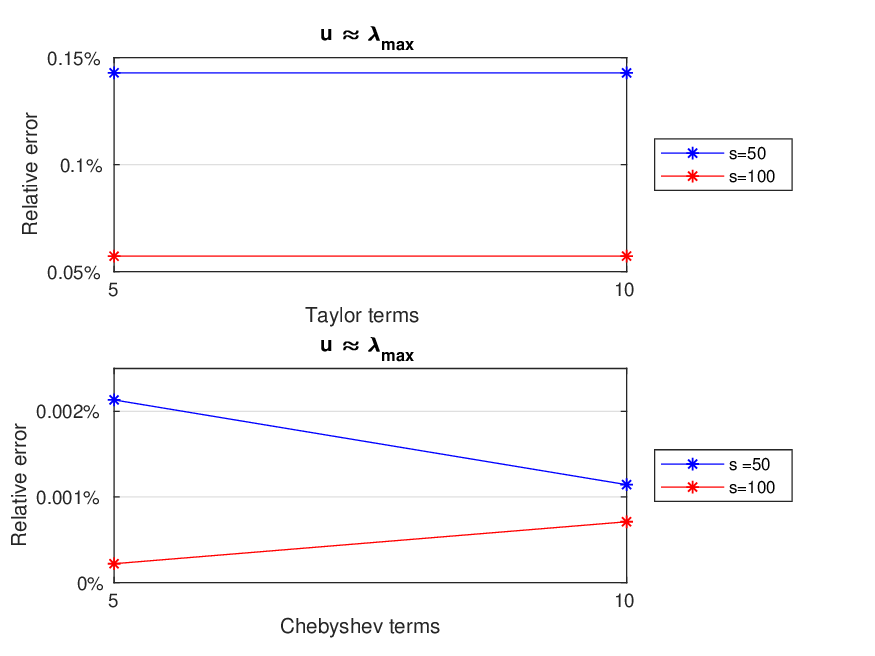}
	\caption{Relative error for the $10^8 \times 10^8$ tridiagonal density matrix using the Taylor and the Chebyshev approximation algorithms with $u={\lambda}_{\max}$.}
	\label{fig:RE10B_banded_smax}
\end{figure}

\begin{figure}[h!]
	\centering
	\includegraphics[width=0.7\textwidth]{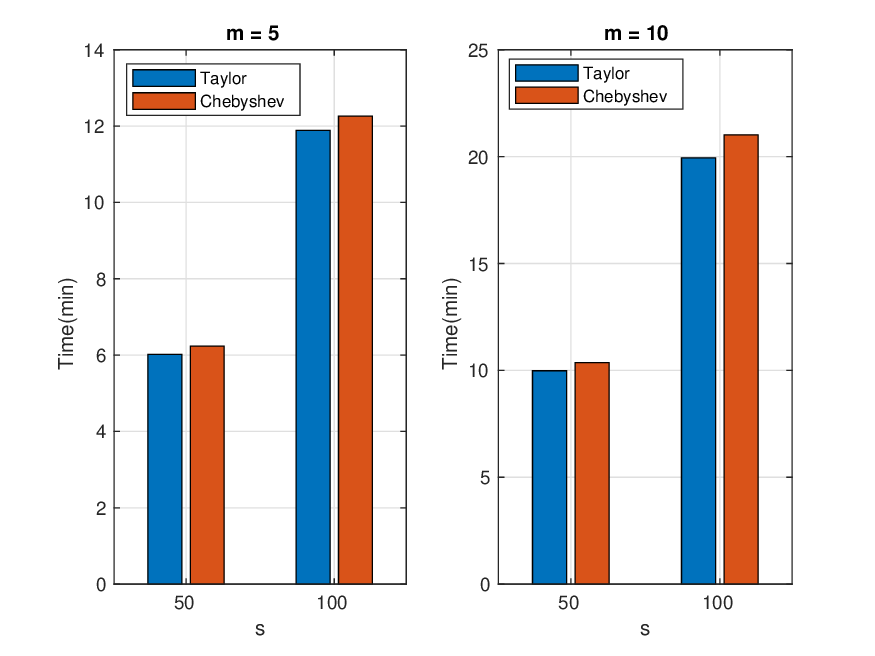}
	\caption{Wall-clock times: Taylor approximation (blue) and Chebyshev approximation (red) for the $10^8 \times 10^8$ triadiagonal density matrix. Exact computation using the Singular Value Decomposition was infeasible using our computational resources. }\label{fig:Time10B_banded_smax}
\end{figure}

The fourth dataset we experimented with includes $5,000 \times 5,000$ density matrices whose first top-$k$ eigenvalues follow a linear decay and the remaining $5,000-k$ a uniform distribution. Let $k$, the number of eigenvalues that follow the linear decay, take values in the set $\{50,\ 1000,\ 3500,\ 5000\}$. Let $m$, the number of terms retained in the Taylor series approximation or the degree of the polynomial used in the Chebyshev polynomial approximation, range between five and 30 in increments of five. Let $s$, the number of random Gaussian vectors used to estimate the trace, be set to $\{50, 100, 200, 300\}$. The estimate of the largest eigenvalue $u$ is set to $\tilde{\lambda}_{\max}$. Figures~\ref{fig:RE5K_nosmooth_d50} to \ref{fig:RE5K_nosmooth_d5000} show the relative error (out of 100\%) for all combinations of $k$, $m$, $s$, and $u$ for the Taylor and Chebyshev approximation algorithms.

We observe that the relative error is decreasing as $k$ increases. It is worth noting that when $k=3,500$ and $k=5,000$ the Taylor-polynomial-based algorithm returns better relative error approximation than the Chebyshev-polynomial-based algorithm. In the latter case we observe that the relative error of the Taylor-based algorithm is almost zero. This observation has a simple explanation. Figure~\ref{fig:eigenvaluedist} shows the distribution of the eigenvalues in the four cases we examine. We observe that for $k=50$ the eigenvalues are spread in the interval $(10^{-2},10^{-4})$; for $k=1,000$ the eigenvalues are spread in the interval $(10^{-3},10^{-4})$; while for $k=3,500$ or $k=5,000$ the eigenvalues are of order $10^{-4}$. It is well known that the Taylor polynomial returns highly accurate approximations when it is computed on values lying inside the open disc centered at a specific value $u$, which, in our case, is the approximation to the dominant eigenvalue. The radius of the disk is roughly $r = \lambda_{m+1}/\lambda_{m}$, where $m$ is the degree of the Taylor polynomial. If $r\leq 1$ then the Taylor polynomial converges; otherwise it diverges. Figure~\ref{fig:convrate} shows the convergence rate for various values of $k$. We observe that for $k=50$ the polynomial diverges, which leads to increased errors for the Taylor-based approximation algorithm (reported error close to $23\%$). In all other cases, the convergence rate is close to one, resulting in negligible impact to the overall error.

In all four cases, the Chebyshev-polynomial based algorithm behaves better or similar to the Taylor-polynomial based algorithm. It is worth noting that when the majority of the eigenvalues are clustered around the smallest eigenvalue, then to achieve relative error similar to the one observed for the QETLAB random density matrices, more than $30$ polynomial terms need to be retained, which increases the computational time of our algorithms. The increase of the computational time as well as the increased relative error can be justified by the large condition number that these matrices have (remember that for both approximation algorithms the running time depends on the approximate condition number $u/l$). As an example, for $k=50$, the condition number is in the order of hundreds which is significant larger than the roughly constant condition number when $k=5,000$.

\begin{figure}[!h]
	\centering
	\includegraphics[width=0.7\textwidth]{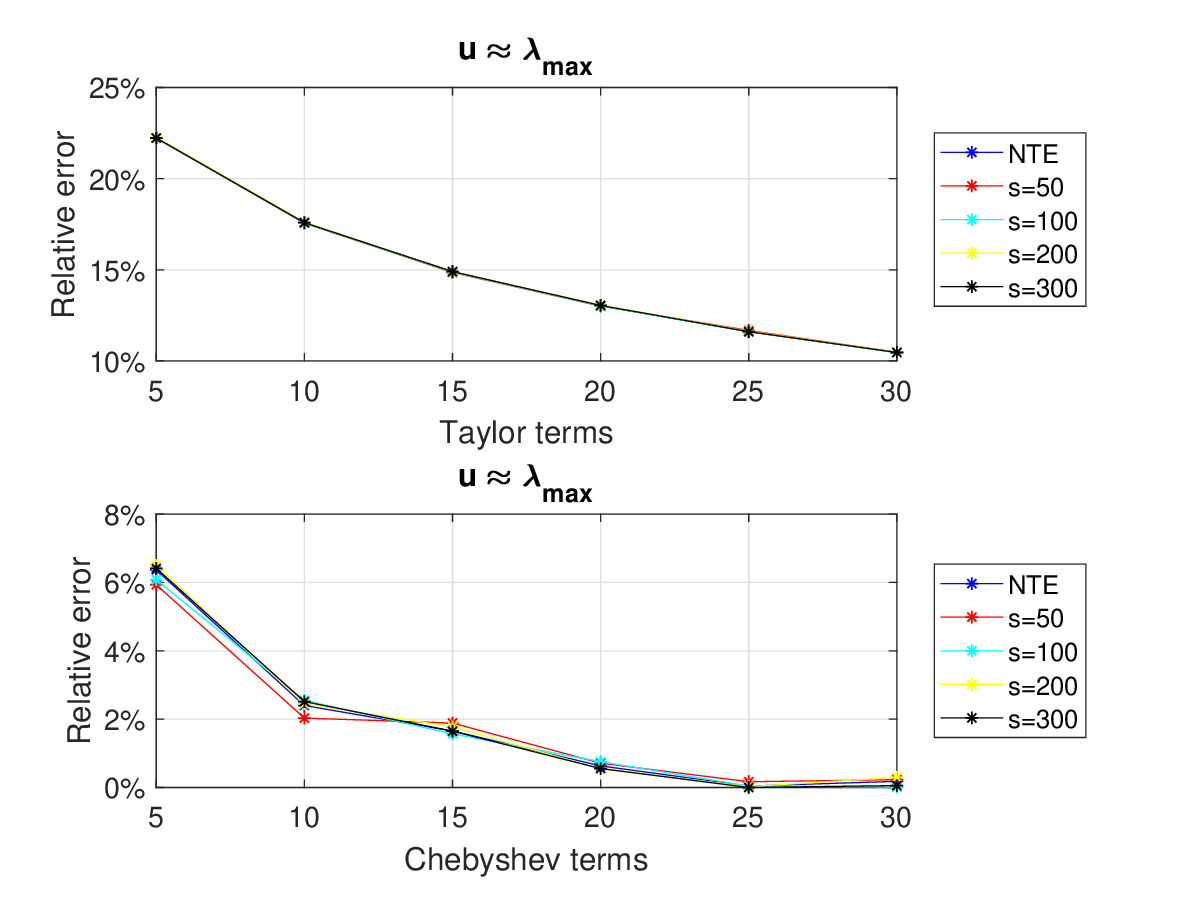}
	\caption{Relative error for $5,000 \times 5,000$ density matrix with the top-$50$ eigenvalues decaying linearly using the Taylor and the Chebyshev approximation algorithms with $u={\lambda}_{\max}$.}
	\label{fig:RE5K_nosmooth_d50}	
\end{figure}

\begin{figure}[!h]
	\centering
	\includegraphics[width=0.7\textwidth]{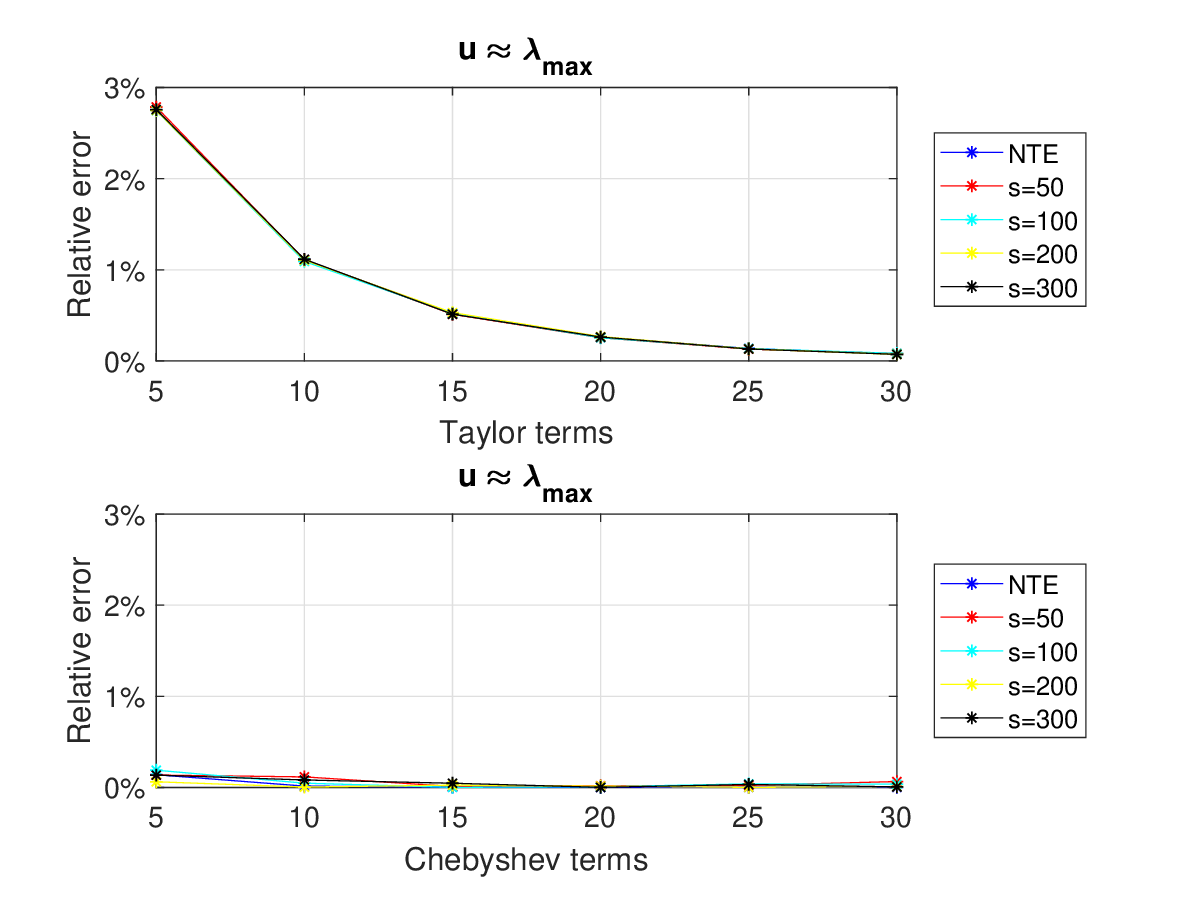}
	\caption{Relative error for $5,000 \times 5,000$ density matrix with the top-$1000$ eigenvalues decaying linearly using the Taylor and the Chebyshev approximation algorithms with $u={\lambda}_{\max}$.}
	\label{fig:RE5K_nosmooth_d1000}
\end{figure}

\begin{figure}[!h]
	\centering
	\includegraphics[width=0.7\textwidth]{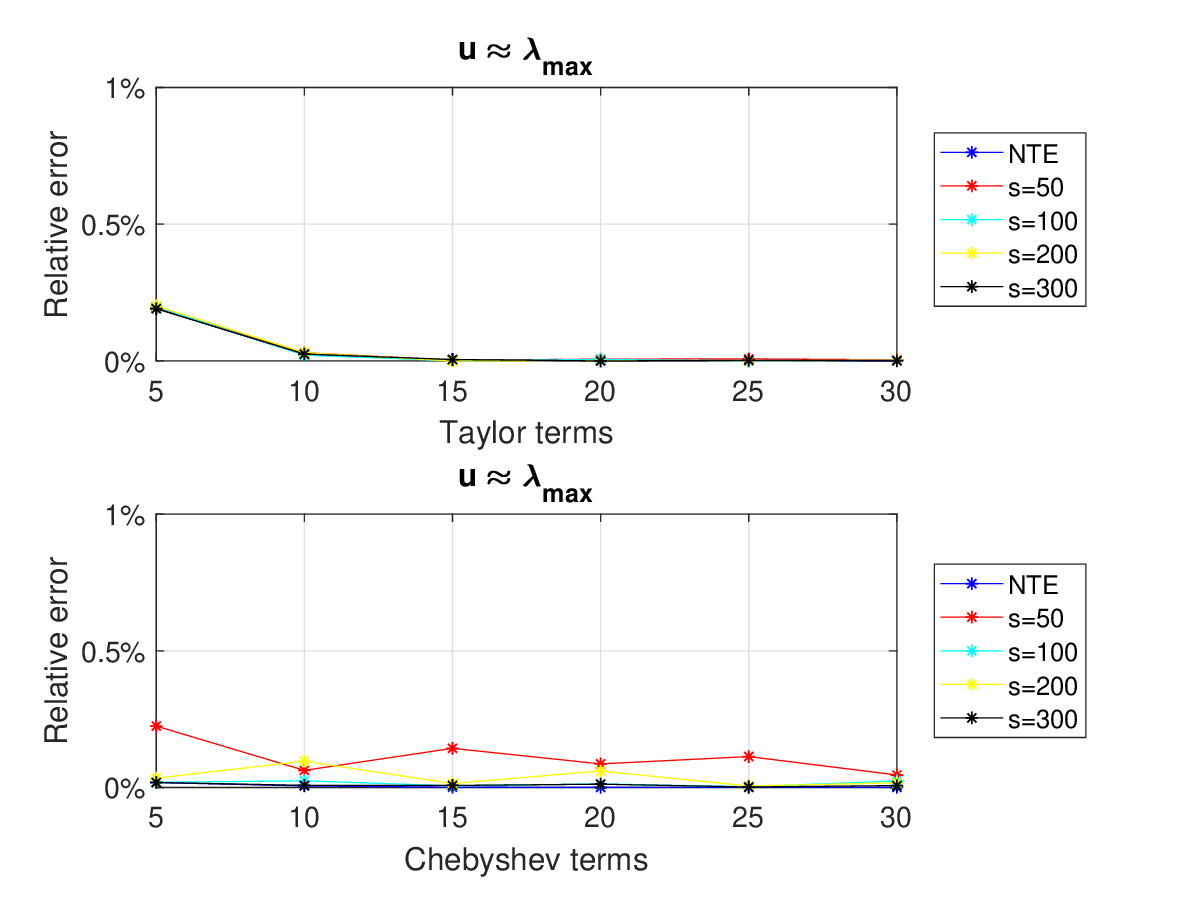}
	\caption{Relative error for $5,000 \times 5,000$ density matrix with the top-$3500$ eigenvalues decaying linearly using the Taylor and the Chebyshev approximation algorithms with $u={\lambda}_{\max}$.}\label{fig:RE5K_nosmooth_d3500}
\end{figure}

\begin{figure}[!h]
\centering
\includegraphics[width=0.7\textwidth]{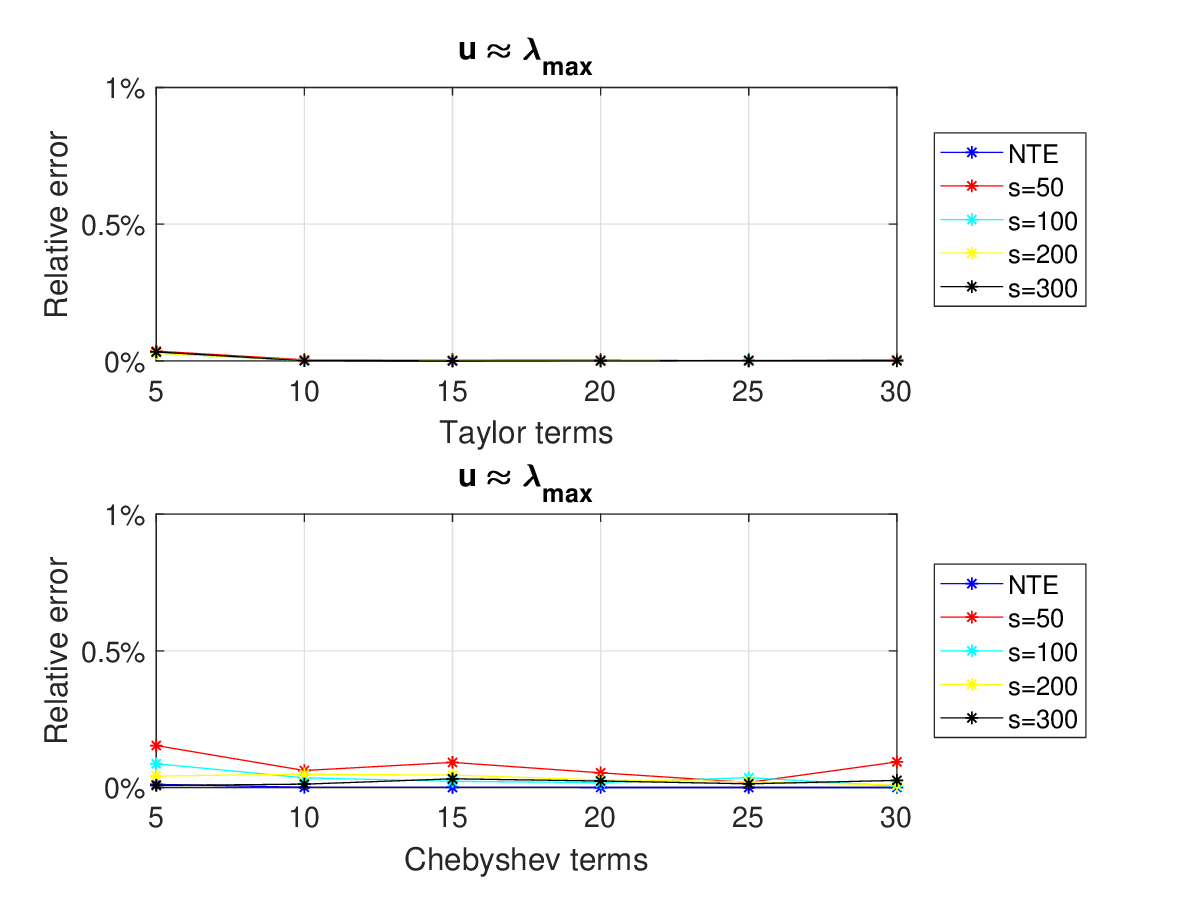}
\caption{Relative error for $5,000 \times 5,000$ density matrix with the top-$5000$ eigenvalues decaying linearly using the Taylor and the Chebyshev approximation algorithms with $u={\lambda}_{\max}$. }\label{fig:RE5K_nosmooth_d5000}
\end{figure}

\begin{figure}[!h]
\centering
\includegraphics[width=0.7\textwidth]{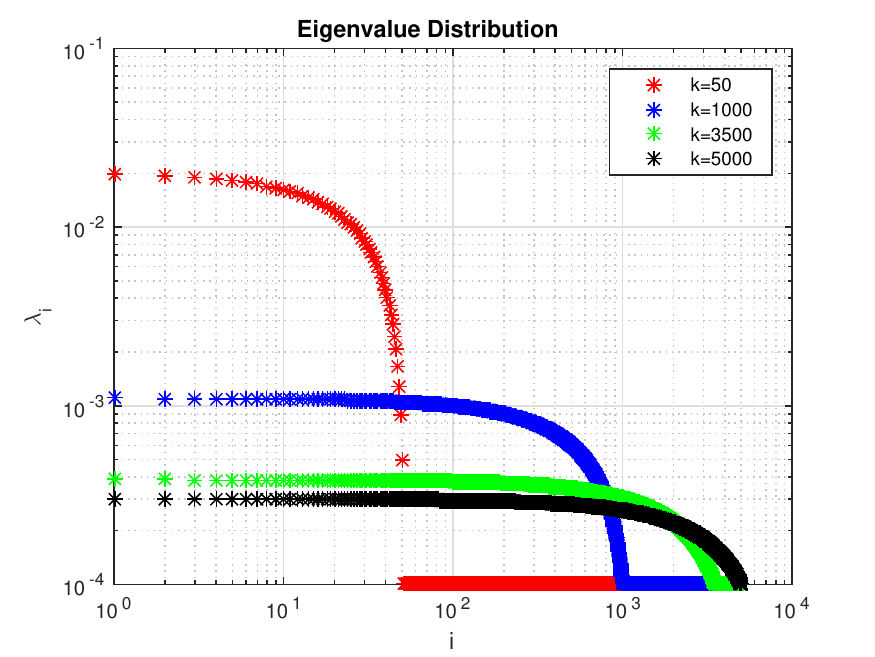}
\caption{Eigenvalue distribution of $5,000 \times 5,000$ density matrices with the top-$k=\{50,\ 1000,\ 3500,\ 5000\}$ eigenvalues decaying linearly and the remaining ones ($5,000-k$) following a uniform distribution.}\label{fig:eigenvaluedist}
\end{figure}

\begin{figure}[!h]
	\centering
	\includegraphics[width=0.7\textwidth]{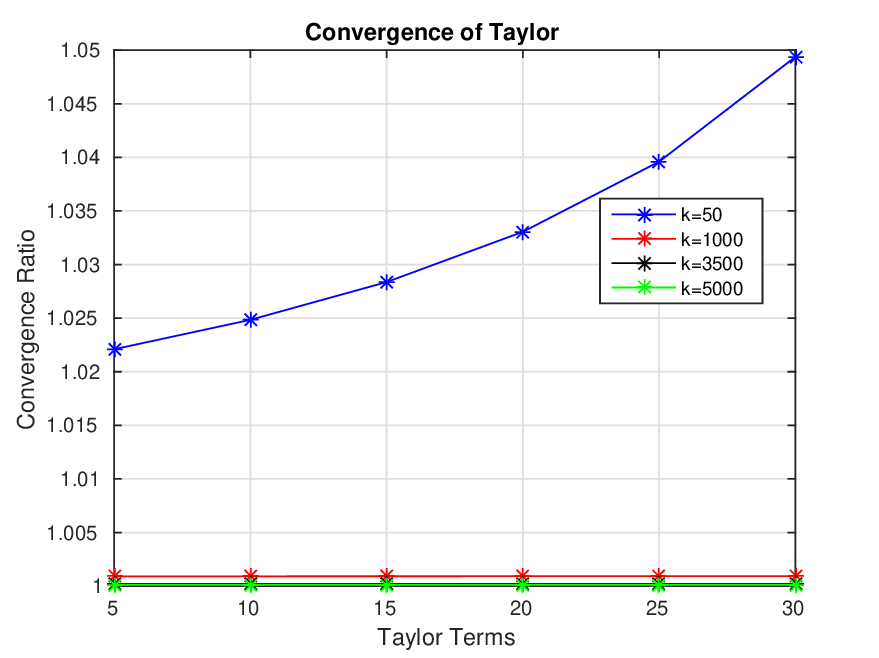}
	\caption{Convergence radius of the Taylor polynomial for the $5,000 \times 5,000$ density matrices with the top-$k=\{50,\ 1000,\ 3500,\ 5000\} $ eigenvalues decaying linearly and the remaining ones ($5,000-k$) following a uniform distribution.}\label{fig:convrate}
\end{figure}

\subsection{Empirical Results for the Hermitian Case}
Our last dataset is a random $5,000 \times 5,000$ complex density matrix generated using the QETLAB Matlab toolbox. We used the function~\texttt{RandomDensityMatrix} of QETLAB and the Haar measure. Let $m$ (the number of terms retained in the Taylor series approximation or the degree of the polynomial used in the Chebyshev polynomial approximation) range between five and 30 in increments of five. Let $s$, the number of random Gaussian vectors used to estimate the trace, be set to $\{50, 100, 200, 300\}$. Figures~\ref{fig:RE5KTaylor_smax_complex} and~\ref{fig:RE5KCheb_smax_complex} show the relative error (out of 100\%) for all combinations of $m$, $s$, and $u$ for the Taylor-based and Chebyshev-based approximation algorithms respectively.

We observe that the relative error is always small, typically below $1\%$, for any choice of the parameters $s$ and $m$. The NTE line (no trace estimation) in the plots serves as a lower bound for the relative error.
We note that computing the exact Von-Neumann entropy took approximately $52$ seconds for matrices of this size. Finally, our algorithm seems to outperform exact computation of the von-Neumann entropy by approximating it in about ten seconds (for the Taylor-based approach) with a relative error of $0.5\%$ using $100$ random Gaussian vectors and retaining ten Taylor terms (see Fig.~\ref{fig:time5KTaylor_complex}) or in about $18$ seconds (for the Chebyshev-based approach) with a relative error of $0.2\%$ using $50$ random Gaussian vectors and five Chebyshev polynomials (see Fig.\ref{fig:time5KCheb_complex}) .

\begin{figure}[h!]
	\centering
	\includegraphics[width=0.7\textwidth]{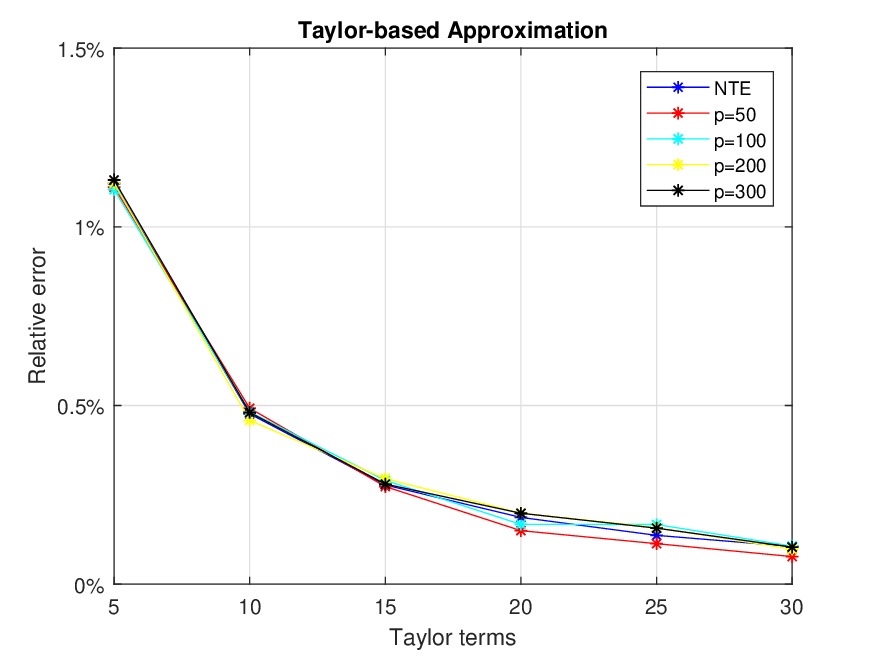}
	\caption{Relative error for $5,000 \times 5,000$ density matrix using the Taylor approximation algorithm.}
	\label{fig:RE5KTaylor_smax_complex}	
\end{figure}

\begin{figure}[h!]
	\centering
	\includegraphics[width=0.7\textwidth]{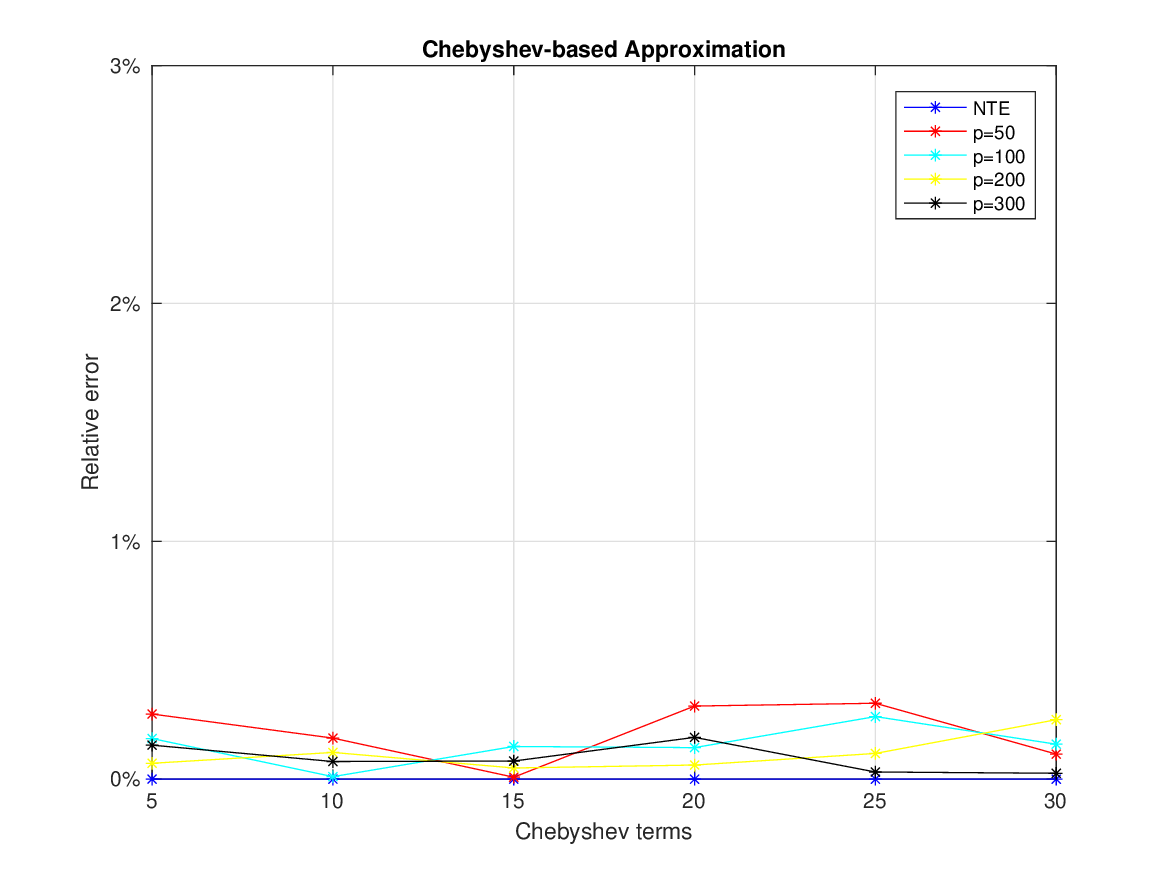}
	\caption{Relative error for $5,000 \times 5,000$ density matrix using the Chebyshev approximation algorithm.}
	\label{fig:RE5KCheb_smax_complex}	
\end{figure}

\begin{figure}[h!]
	\centering
	\includegraphics[width=0.7\textwidth]{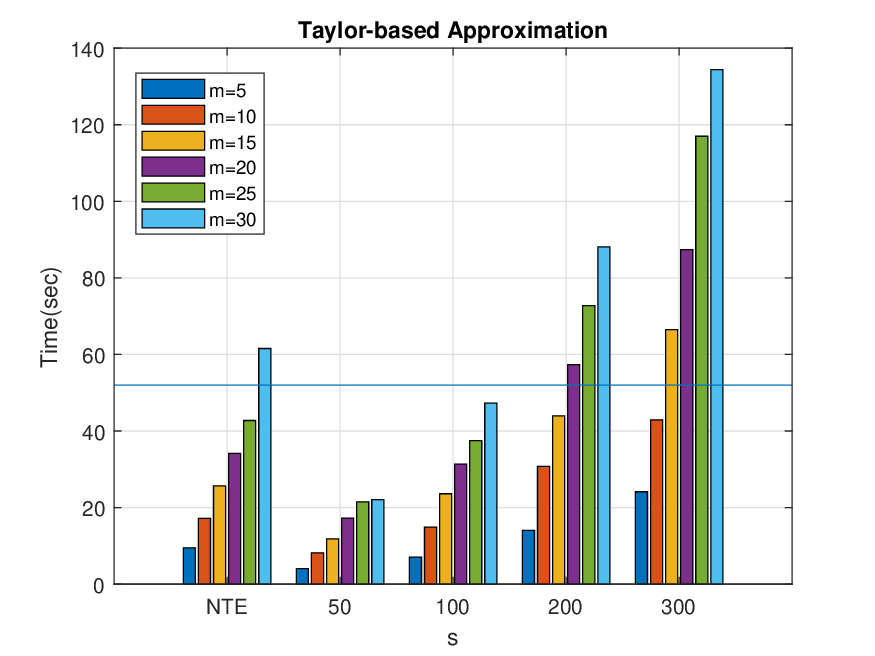}
	\caption{Time (in seconds) to run the Taylor-based algorithm for the $5,000 \times 5,000$ density matrix for all combinations of $m$ and $s$. \textit{Exactly} computing the Von-Neumann entropy took approximately 52 seconds, designated by the straight horizontal line in the figure.}\label{fig:time5KTaylor_complex}
\end{figure}

\begin{figure}[h!]
	\centering
	\includegraphics[width=0.7\textwidth]{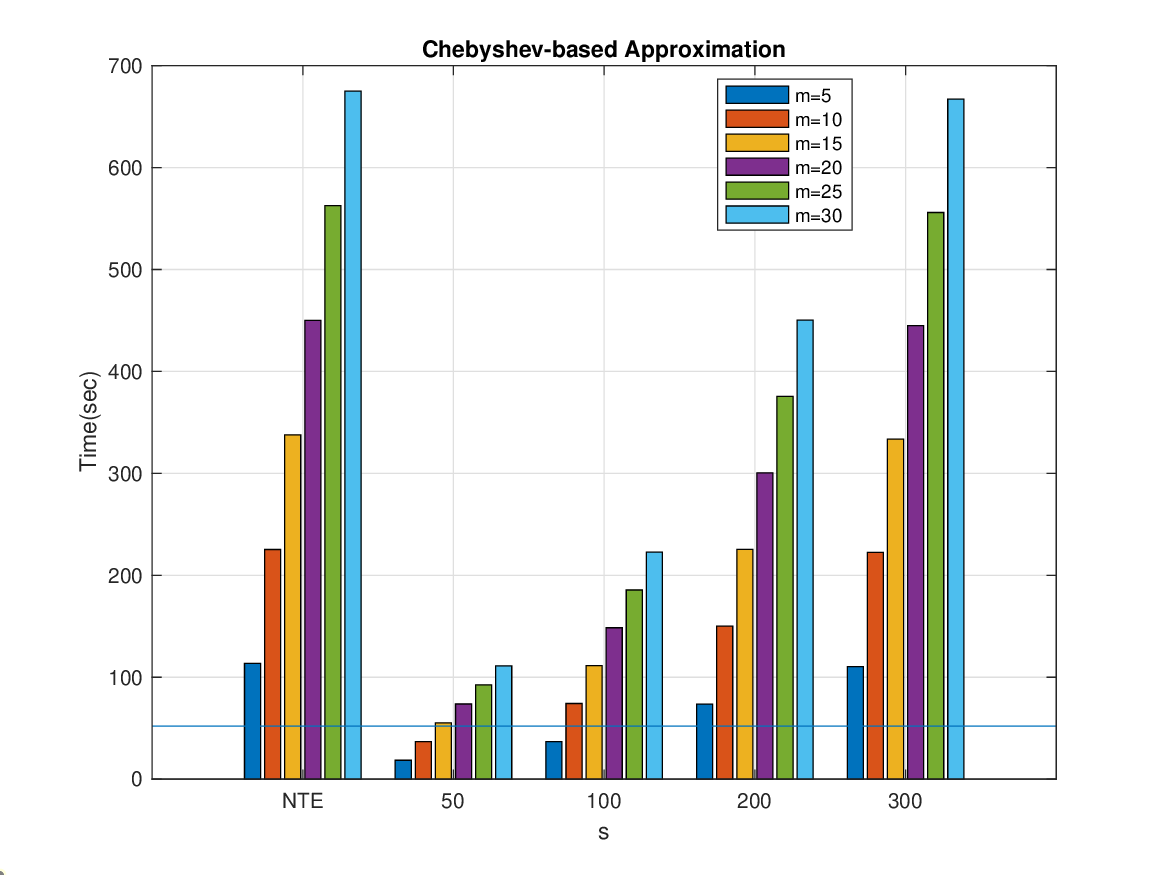}
	\caption{Time (in seconds) to run the Chebyshev-based algorithm for the $5,000 \times 5,000$ density matrix for all combinations of $m$ and $s$. \textit{Exactly} computing the Von-Neumann entropy took approximately 52 seconds, designated by the straight horizontal line in the figure.}\label{fig:time5KCheb_complex}
\end{figure}

\subsection{Empirical results for the random projection approximation algorithms}
	
In order to evaluate our third algorithm, we generated low-rank random density matrices (recall that the algorithm of Section~\ref{sxn:rp} works only for random density matrices of rank $k$ with $k \ll n$). Additionally, in order to evaluate the subsampled randomized Hadamard transform and avoid padding with all-zero rows, we focused on values of $n$ (the number of rows and columns of the density matrix) that are powers of two. Finally, we also evaluated a simpler random projection matrix, namely the Gaussian random matrix, whose entries are all Gaussian random variables with zero mean and unit variance.

We generated \textit{low rank} random density matrices with exponentially (using the QETLAB Matlab toolbox) and linearly decaying eigenvalues. The sizes of the density matrices we tested were $4,096 \times 4,096$ and $16,384 \times 16,384$. We also generated much larger $30,000 \times 30,000$ random matrices on which we only experimented with the Gaussian random projection matrix.

We computed all the non-zero singular values of a matrix using the~\texttt{svds} function of Matlab in order to take advantage of the fact that the target density matrix has low rank. The accuracy of our proposed approximation algorithms was evaluated by measuring the relative error; wall-clock times were reported in order to quantify the speedup that our approximation algorithms were able to achieve.

We start by reporting results for Algorithm~\ref{alg:EN} using the Gaussian, the subsampled randomized Hadamard transform (Algorithm~\ref{alg:SHD}), and the input-sparsity transform (Algorithm~\ref{alg:IS}) random projection matrices. Consider the $4,096 \times 4,096$ low rank density matrices and let $k$, the rank of the matrix, be $10$, $50$, $100$, and $300$. Let $s$, the number of columns of the random projection matrix, range from $50$ to $1,000$ in increments of $50$. Figures~\ref{fig:RP_RE_4K} and~\ref{fig:RP_RE_ns_4K} depict the relative error (out of $100\%$) for all combinations of $k$ and $s$. We also report the wall-clock running times for values of $s$ between $300$ and $450$ at Figure~\ref{fig:RP_TM_4K}.

We observe that in the case of the random matrix with exponentially decaying eigenvalues and for all algorithms the relative error is under $0.3\%$ for any choice of the parameters $k$ and $s$ and, as expected, decreases as the dimension of the projection space $s$ grows larger. Interestingly, all three random projection matrices returned essentially identical accuracies and very comparable wall-clock running time results. This observation is due to the fact that for all choices of $k$, after scaling the matrix to unit trace, the only eigenvalues that were numerically non-zero were the $10$ dominant ones.
	
In the case of the random matrix with linearly decaying eigenvalues (and for all algorithms) the relative error increases as the rank of the matrix increases and decreases as the size of the random projection matrix increases. This is expected: as the rank of the matrix increases, a larger random projection space is needed to capture the ``energy'' of the matrix. Indeed, we observe that for all values of $k$, setting $s=1,000$ guarantees a relative error under $1\%$. Similarly, for $k=10$, the relative error is under $0.3\%$ for any choice of $s$.

The running time depends not only on the size of the matrix, but also on its rank, e.g. for $k=100$ and $s=450$, our approximation was computed in about $2.5$ seconds, whereas for $k=300$ and $s=450$, it was computed in less than one second. Considering, for example, the case of $k=300$ exponentially decaying eigenvalues, we observe that for $s=400$ we achieve relative error below $0.15$\% and a speedup of over $60$ times compared to the exact computation. Finally, it is observed that all three algorithms returned very comparable wall-clock running time results. This observation could be due to the fact that matrix multiplication is heavily optimized in Matlab and therefore the theoretical advantages of the Hadamard transform did not manifest themselves in practice.

\begin{figure}[!h]
	\centering
	\includegraphics[width=0.7\textwidth]{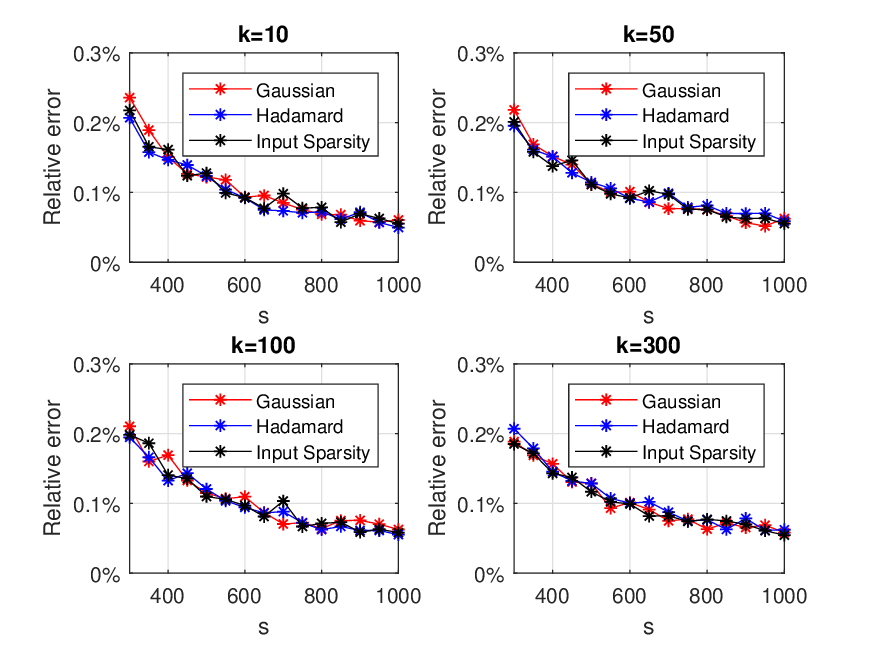}
	\caption{Relative error for the $4,096 \times 4,096$ rank-$k$ density matrix with exponentially decaying eigenvalues using Algorithm~\ref{alg:EN} with the Gaussian (red), the subsampled randomized Hadamard transform (blue), and the input sparsity transform (black) random projection matrices.}
	\label{fig:RP_RE_4K}	
\end{figure}

\begin{figure}[!h]
	\centering
	\includegraphics[width=0.7\textwidth]{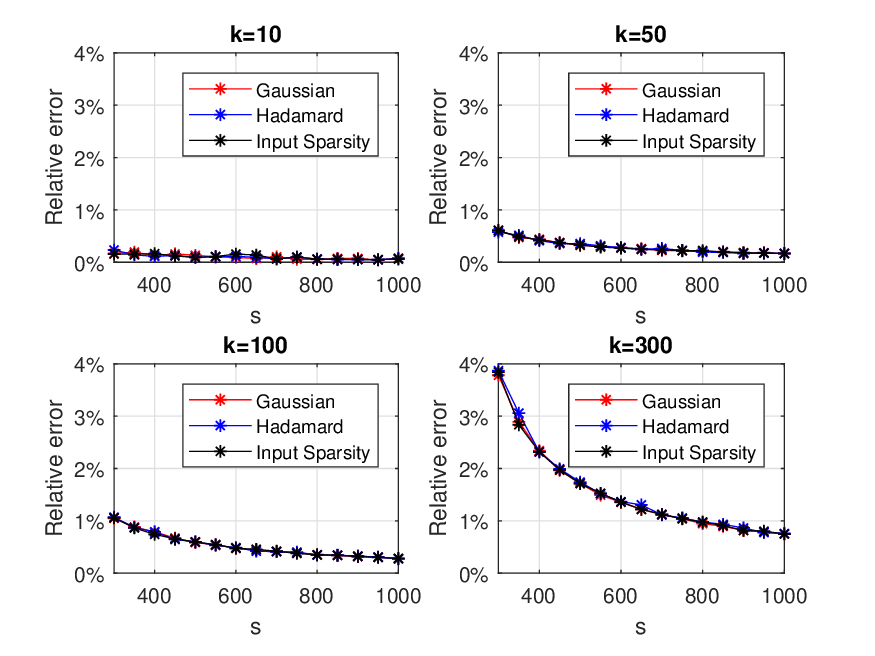}
	\caption{Relative error for the $4,096 \times 4,096$ rank-$k$ density matrix with linearly decaying eigenvalues using Algorithm~\ref{alg:EN} with the Gaussian (red), the subsampled randomized Hadamard transform (blue), and the input sparsity transform (black) random projection matrices.}
	\label{fig:RP_RE_ns_4K}	
\end{figure}

\begin{figure}[!h]
	\centering
	\includegraphics[width=0.7\textwidth]{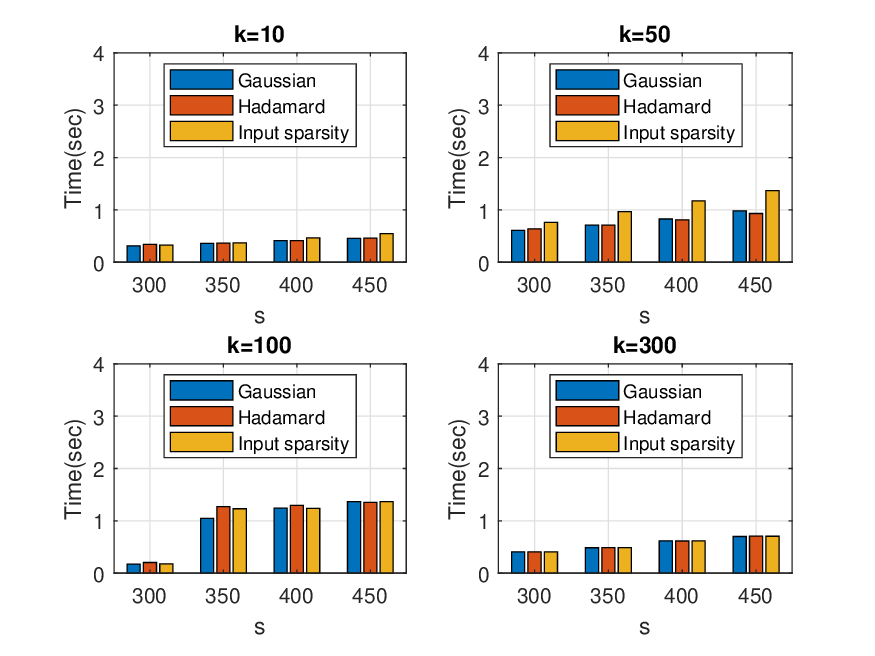}
	\caption{Wall-clock times: Algorithm~\ref{alg:EN} on $4,096 \times 4,096$ random matrices, with the Gaussian (blue), the subsampled randomized Hadamard transform (red) and the input sparsity transform (orange) projection matrices. The exact entropy was computed in $1.5$ seconds for the rank-10 approximation, in eight seconds for the rank-50 approximation, in $15$ seconds for the rank-100 approximation, and in one minute for the rank-300 approximation.}
	\label{fig:RP_TM_4K}
\end{figure}

The second dataset we experimented with was a $16,384 \times 16,384$ low rank density matrix. We set  $k=50$ and $k=500$ and we let $s$  take values in the set $\{500,\ 1000,\ 1500,\ldots,3000,\ 3500\}$. We report the relative error (out of $100\%$) for all combinations of $k$ and $s$ in Figure~\ref{fig:RP_RE_16K} for the matrix with exponentially decaying eigenvalues and in Figure~\ref{fig:RP_RE_16K_ns} for the matrix with linearly decaying eigenvalues. We also report the wall-clock running times for $s$ between $500$ and $2,000$ in Figure~\ref{fig:RP_TM_16K}.
We observe that the relative error is typically around $1\%$ for both types of matrices, with running times ranging between ten seconds and four minutes, significantly outperforming the exact entropy computation which took approximately $1.6$ minutes for the rank 50 approximation and 20 minutes for the rank 500 approximation.

\begin{figure}[!h]
	\centering
	\includegraphics[width=0.7\textwidth]{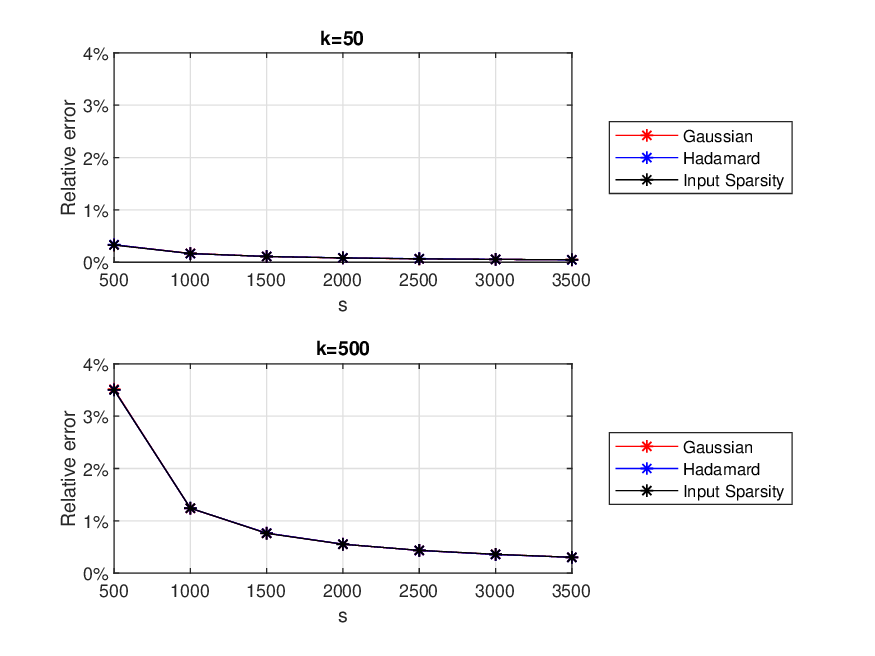}
	\caption{Relative error for the $16,384 \times 16,384$ rank-$k$ density matrix with exponentially decaying eigenvalues using Algorithm~\ref{alg:EN} with the Gaussian (red), the subsampled randomized Hadamard transform (blue), and the input sparsity transform (black) random projection matrices.}
	\label{fig:RP_RE_16K}	
\end{figure}

\begin{figure}[h!]
	\centering
	\includegraphics[width=0.7\textwidth]{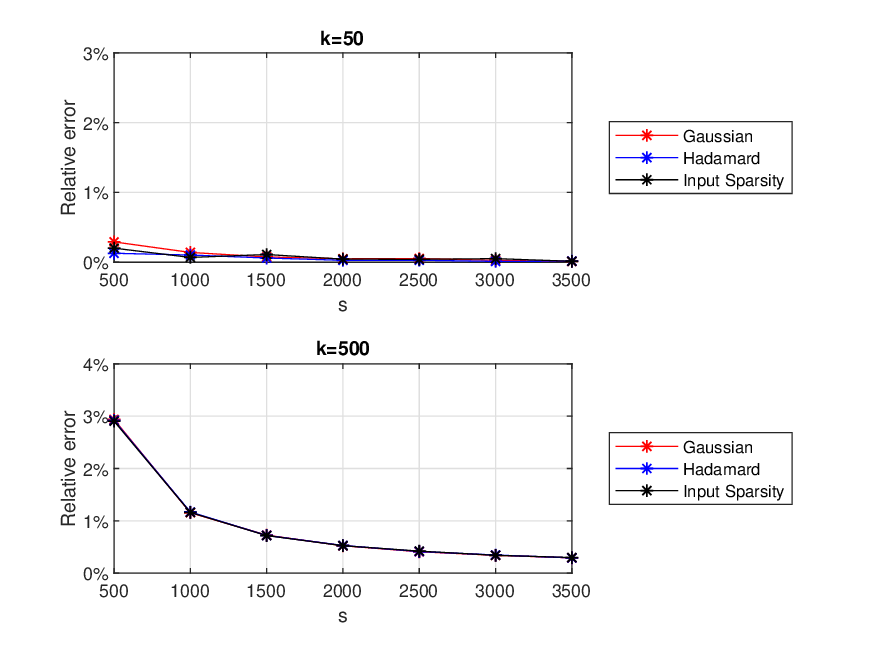}
	\caption{Relative error for the $16,384 \times 16,384$ rank-$k$ density matrix with linearly decaying eigenvalues using Algorithm~\ref{alg:EN} with the Gaussian (red), the subsampled randomized Hadamard transform (blue), and the input sparsity transform (black) random projection matrices.}
	\label{fig:RP_RE_16K_ns}	
\end{figure}

\begin{figure}[h!]
	\centering
	\includegraphics[width=0.7\textwidth]{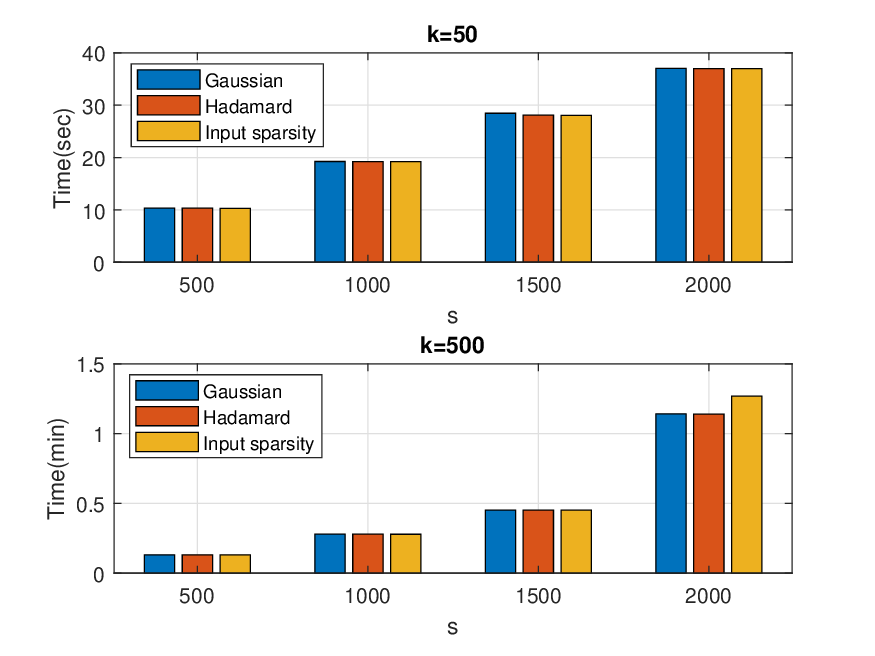}
	\caption{Wall-clock times: Algorithm~\ref{alg:EN} with the Gaussian (blue), the subsampled randomized Hadamard transform (red) and the input sparsity transform (orange) projection matrices. The exact entropy was computed in $1.6$ minutes for the rank 50 approximation and in $20$ minutes for the rank 500 approximation.}
	\label{fig:RP_TM_16K}
\end{figure}

The last dataset we experimented with was a $30,000 \times 30,000$ low rank density matrix on which we ran Algorithm~\ref{alg:EN} using a Gaussian random projection matrix. We set $k=50$ and $k=500$ and we let $s$ take values in the set $\{500,\ 1000,\ 1500,\ldots,3000,\ 3500\}$. We report the relative error (out of $100\%$) for all combinations of $k$ and $s$ in Figure~\ref{fig:RP_RE_30K} for the matrix with exponentially decaying eigenvalues and in Figure~\ref{fig:RP_RE_30K_ns} for the matrix with the linearly decaying eigenvalues. We also report the wall-clock running times for $s$ ranging between $500$ and $2,000$ in Figure~\ref{fig:RP_TM_30K}. We observe that the relative error is typically around $1\%$ for both types of matrices, with the running times ranging between $30$ seconds and two minutes, outperforming the exact entropy which was computed in six minutes for the rank 50 approximation and in one hour for the rank 500 approximation.

\begin{figure}[h!]
	\centering
	\includegraphics[width=0.7\textwidth]{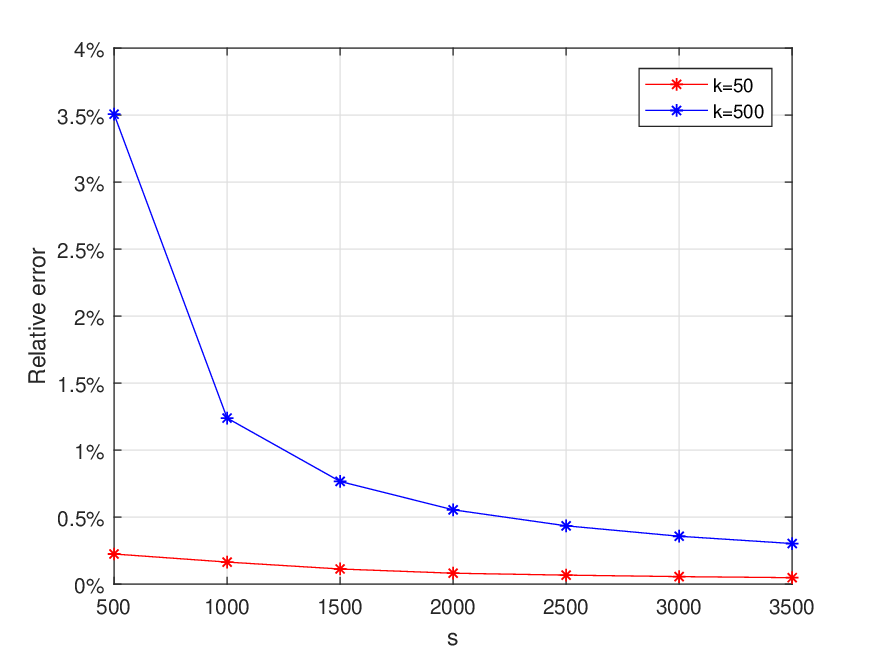}
	\caption{Relative error for the $30,000 \times 30,000$ rank-$k$ density matrix with exponentially decaying eigenvalues using Algorithm~\ref{alg:EN} with the Gaussian random projection matrix for $k=50$ (red) and for $k=500$ (blue).}
	\label{fig:RP_RE_30K}	
\end{figure}

\begin{figure}[h!]
	\centering
	\includegraphics[width=0.7\textwidth]{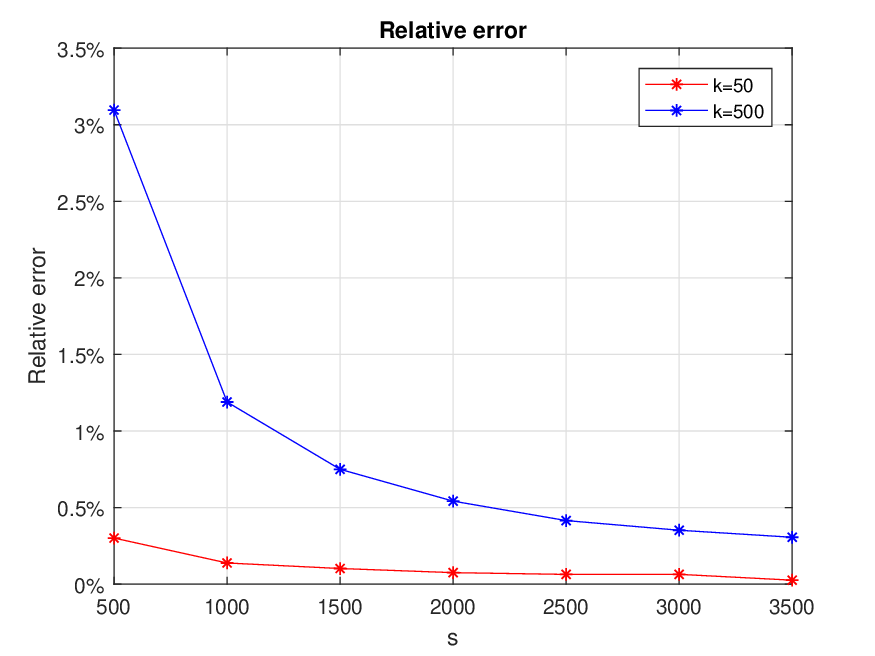}
	\caption{Relative error for the $30,000 \times 30,000$ rank-$k$ density matrix with linearly decaying eigenvalues using Algorithm~\ref{alg:EN} with the Gaussian random projection matrix for $k=50$ (red) and for $k=500$ (blue).}
	\label{fig:RP_RE_30K_ns}	
\end{figure}

\begin{figure}[h!]
	\centering
	\includegraphics[width=0.7\textwidth]{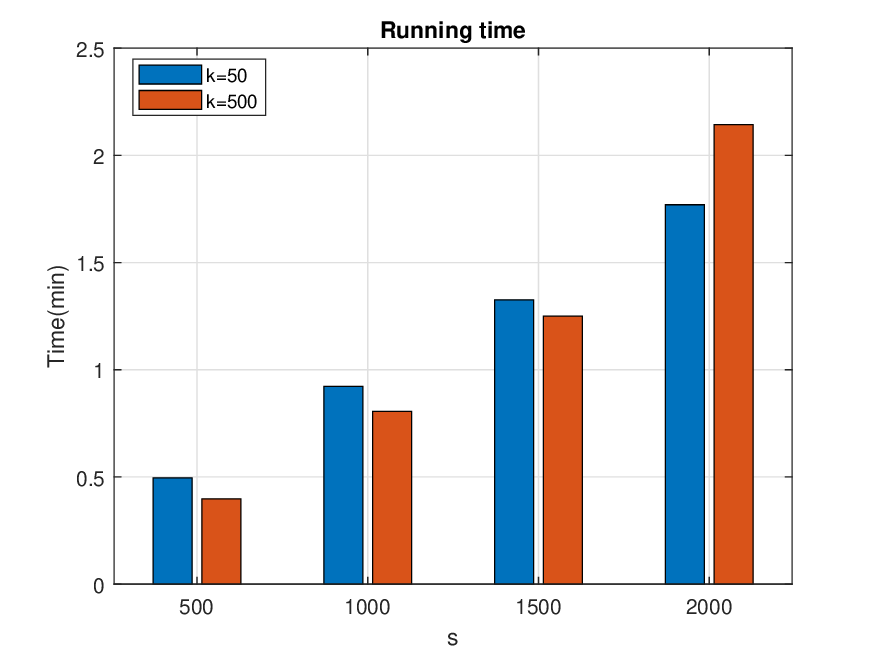}
	\caption{Wall-clock times: rank-50 approximation (blue) and rank-500 approximation (red). Exact computation needed about six minutes and one hour respectively.}
	\label{fig:RP_TM_30K}
\end{figure} 

\section{Conclusions and open problems}

We presented and analyzed three randomized algorithms to approximate the von Neumann entropy of density matrices. Our algorithms leverage recent developments in the RandNLA literature: randomized trace estimators, provable bounds for the power method, the use of random projections to approximate the singular values of a matrix, etc. All three algorithms come with provable accuracy guarantees under assumptions on the spectrum of the density matrix. Empirical evaluations on $30,000 \times 30,000$ synthetic density matrices support our theoretical findings and demonstrate that we can efficiently approximate the von Neumann entropy in a few minutes with minimal loss in accuracy, whereas an the exact computation takes over 5.5 hours.

An interesting open problem would be to consider the estimation of the cross entropy. The cross entropy is a measure between two probability distributions and is particularly important in information theory.  Algebraically, it can be defined as  $\entropy{\bS,\bR} = -\trace{\bS\log{\bR}}$, where $\bS\in \comp^{n\times n}$ and $\bR\in\comp^{n\times n}$ are density matrices with a full set of pure states.  One can further extend our polynomial-based approaches using the Taylor expansion or the Chebyshev polynomials to approximate the matrix $\Gamma = \bS\log\bR$. The case where both or one of the density matrices have an incomplete set of pure states is an open problem: if $\bR$ is low-rank, then our first two approaches would not work for the reasons discussed in Section~\ref{sxn:rp}. However, if the only low rank matrix is $\bS$, then our first two approaches would still work: $\bS$ is only appearing in the trace estimation part, and having eigenvalues equal to zero does not affect the positive semi-definiteness of $\Gamma$. When $\bR$ is of low rank then one might be able to use our random projection approaches to reduce its dimensionality and/or the dimensionality of $\bS$.

The most important open problem is to relax (or eliminate) the assumptions associated with our three key technical results without sacrificing our running time guarantees. It would be critical to understand whether our assumptions are, for example, necessary to achieve relative error approximations and either provide algorithmic results that relax or eliminate our assumptions or provide matching lower bounds and counterexamples.

\bibliographystyle{plain}
\bibliography{entropy}

\newpage \section{The power method}

We consider the well-known power method to estimate the largest eigenvalue of a matrix. In our context, we will use the power method to estimate the largest probability $p_i$ for a density matrix $\bR$.
\renewcommand\labelitemii{$\bullet$}
\begin{algorithm}[htb!]
\begin{itemize}
\item \textbf{INPUT:} SPD matrix $\bA \in \mathbb{R}^{n \times n},$ integers $q,\ t > 0$.
\item For $j=1,\dots,q$
\begin{enumerate}
\item Pick uniformly at random a vector $\bx_0^j \in \{ +1, -1\}^{n}$.
\item For $i=1,\dots,t$
\begin{itemize}
 \item $\bx_i^j = \bA \cdot \bx^j_{i-1}$.
 \end{itemize}
\item Compute: $\tilde{p}_1^j =  \frac{{\bx_t^j}^T \bA \bx_t^j}{{\bx_t^j}^T \bx_t^j}$.
 \end{enumerate}
\item \textbf{OUTPUT:} $\tilde{p}_1 =  \max_{j=1\ldots q} \tilde{p}_1^j$.
\end{itemize}
\caption{Power method, repeated $q$ times.}\label{alg:power}
\end{algorithm}

\noindent Algorithm~\ref{alg:power} requires $\OO(qt (n + \nnz(\bA)))$ arithmetic operations
to compute $\tilde{p}_1$. The following lemma appeared in~\cite{Boutsidis2016}, building upon~\cite{LT_Lecture}.
\begin{lemma}
	\label{lem:power}
Let $\tilde{p}_1$ be the output of Algorithm~\ref{alg:power} with $q=\left\lceil 4.82 \log(1/\delta)\right\rceil$ and $t = \left\lceil\log \sqrt{4n}\right\rceil$. Then, with probability at least $1-\delta$,
$$\frac 1 6 p_1 \le \tilde{p}_1 \leq p_1.$$

The running time of Algorithm~\ref{alg:power} is $\OO\left(\left(n+\nnz(\bA)\right)\log(n)\log\left(\frac{1}{\delta}\right)\right).$
\end{lemma}

\section{The Clenshaw Algorithm}

We briefly sketch Clenshaw's algorithm to evaluate Chebyshev polynomials with matrix inputs. Clenshaw's algorithm is a recursive approach with base cases $b_{m+2}(x) = b_{m+1}(x) = 0$ and the recursive step (for $k = m, m-1, \ldots, 0$):

\begin{equation}\label{eq:clenshaw_recstep}
b_k(x) = \alpha_k + 2xb_{k+1}(x) - b_{k+2}(x).
\end{equation}
(See Section~\ref{sxn:cheb} for the definition of $\alpha_k$.) Then,
\begin{equation}\label{eq:clenshaw_finalsum}
f_m(x) = \frac{1}{2}\left(\alpha_0 + b_0(x) -b_2(x)\right).
\end{equation}
Using the mapping $x\rightarrow 2(x/u)  - 1$, eqn.~\eqref{eq:clenshaw_recstep} becomes
\begin{equation}\label{eq:clenshaw_recstep2}
b_k(x) = \alpha_k + 2\left(\frac{2}{u}x-1\right)b_{k+1}(x) - b_{k+2}(x).
\end{equation}
In the matrix case, we substitute $x$ by a matrix. Therefore, the base cases are $\bB_{m+2}(\bR) = \bB_{m+1}(\bR) = \bzero$ and the recursive step is
\begin{equation}\label{eq:clenshaw_recstep2m}
\bB_k(\bR) = \alpha_k\bI_n + 2\left(\frac{2}{u}\bR-\bI_n\right)\bB_{k+1}(\bR) - \bB_{k+2}(\bR)
\end{equation}
for $k = m, m-1, \ldots, 0$. The final sum is
\begin{equation}\label{eq:clenshaw_finalsum_m}
f_m(\bR) = \frac{1}{2}\left(\alpha_0\bI_n + \bB_0(\bR) -\bB_2(\bR)\right).
\end{equation}
Using the matrix version of Clenshaw's algorithm, we can now rewrite the trace estimation $\bg^\top f_m(\bR)\bg$ as follows. First, we right multiply eqn.~\eqref{eq:clenshaw_recstep2m} by $\bg$,
\begin{eqnarray}
\nonumber \bB_k(\bR)\bg &=& \alpha_k\bI_n\bg + 2\left(\frac{2}{u}\bR-\bI_n\right)\bB_{k+1}(\bR)\bg - \bB_{k+2}(\bR)\bg, \\
\by_k & = & \alpha_k\bg + 2\left(\frac{2}{u}\bR-\bI_n\right)\by_{k+1} -\by_{k+2}.\label{eq:fmRg}
\end{eqnarray}
Eqn.~\eqref{eq:fmRg} follows by substituting $\by_i = \bB_i(\bR)\bg$. Multiplying the base cases by $\bg$, we get $\by_{m+2} = \by_{m+1} =\bzero$ and the final sum becomes
\begin{equation}\label{eq:fmRg_final}
\bg^\top f_m(\bR) \bg = \frac{1}{2}\left(\alpha_0(\bg^\top\bg) +
\bg^\top(\by_0 - \by_2)\right).
\end{equation}
Algorithm~\ref{alg_clenshaw_mv} summarizes all the above.

\begin{algorithm}[h!]
	\begin{algorithmic}[1]
		\STATE {\bf{INPUT}}: $\alpha_i$,\ $i = 0,\dots,m$, $\bR\in \real^{n\times n}$, $\bg \in \real^{n}$
		\STATE Set 	$\by_{m+2} =\by_{m+1} = \bzero$
		\FOR {$k=m, m-1, \ldots ,0$ }
		\STATE $\by_k = \alpha_k\bg + \frac{4}{u}\bR\by_{k+1} -  2\by_{k+1}-\by_{k+2}$
		\ENDFOR
        \STATE{\bf{OUTPUT}: $\bg^\top f_m(\bR) \bg =\frac{1}{2}\left(\alpha_0(\bg^\top\bg) + \bg^\top(\by_0 - \by_2)\right)$}
	\end{algorithmic}
	\caption{{Clenshaw's algorithm to compute $\bg^\top f_m(\bR) \bg$.}}
	\label{alg_clenshaw_mv}
\end{algorithm}

\end{document}